\documentclass[letterpaper,12pt,oneside,reqno]{amsart}
\usepackage{amssymb,amsthm,amsfonts,amsmath, amscd}

\usepackage{hyperref}
\usepackage{mathrsfs}
\usepackage{color}
\usepackage{euscript}

\usepackage[utf8]{inputenc}
\usepackage[english]{babel}
\usepackage{hyperref}
\usepackage{graphicx}
\usepackage{enumerate}
\usepackage{tikz}
\usepackage[DIV14]{typearea}
\usepackage[width=.9\textwidth]{caption}
\allowdisplaybreaks

\usepackage{tikz}
\usetikzlibrary{decorations.markings, arrows}

\renewcommand{\i}{\mathbf{i}}
\DeclareMathOperator{\Prob}{\mathsf{Prob}}
\newcommand{\HT}{\mathfrak{h}}
\newcommand{\SM}{\ensuremath{\mathbf{SM}}}
\newcommand{\Sym}{\ensuremath{\mathrm{Sym}}}
\newcommand{\vv}{\ensuremath{\mathbf{6v}}}
\newcommand{\LL}{\mathcal{L}^{(q)}}
\renewcommand\l{\mathfrak l}
\renewcommand\r{\mathfrak r}

\newcommand\Y{\mathbb Y}
\newcommand\Z{\mathbb Z}
\newcommand\C{\mathbb C}
\newcommand\R{\mathbb R}

\newcommand\E{\mathbb E}

\renewcommand\P{\mathbb P}

\newcommand\al{\alpha}
\newcommand\be{\beta}
\newcommand\ga{\gamma}
\newcommand\Ga{\Gamma}
\newcommand\de{\delta}

\newcommand\la{\lambda}
\newcommand\si{\sigma}
\renewcommand\th{\theta}
\newcommand\epsi{\varepsilon}

\newcommand\X{\mathfrak X}

\newcommand\wt{\widetilde}
\newcommand\wh{\widehat}

\newcommand\Conf{\operatorname{Conf}}

\newcommand\Charlier{{\operatorname{Charlier}}}
\newcommand\dHermite{{\operatorname{dHermite}}}

\newcommand\Lag{{\operatorname{Laguerre}}}

\newcommand\Meixner{{\operatorname{Meixner}}}
\newcommand\Hahn{{\operatorname{Hahn}}}
\newcommand\Racah{{\operatorname{Racah}}}
\newcommand\Krawtchouk{{\operatorname{Krawtchouk}}}

\newcommand\const{\operatorname{const}}
\newcommand\supp{\operatorname{supp}}

\renewcommand\DH{{\operatorname{DHermite}}}
\newcommand\DL{{\operatorname{DLaguerre}}}
\renewcommand\DJ{{\operatorname{DJacobi}}}
\newcommand\Airy{{\operatorname{Airy}}}

\newcommand\ab{{(a,b)}}
\newcommand\alde{{\al,\be,\ga,\de}}
\renewcommand\r{r}

\renewcommand\H{\mathcal H}
\newcommand\W{\mathcal W}
\newcommand\D{\mathscr D}
\newcommand\A{\mathcal A}
\newcommand\J{\mathscr J}
\newcommand\PP{\mathcal P}

\newtheorem{theorem}{Theorem}[section]
\newtheorem*{theorem*}{Theorem}
\newtheorem{proposition}[theorem]{Proposition}
\newtheorem{lemma}[theorem]{Lemma}
\newtheorem{corollary}[theorem]{Corollary}

\theoremstyle{definition}
\newtheorem{definition}[theorem]{Definition}
\newtheorem{remark}[theorem]{Remark}

\numberwithin{equation}{section}

\begin{document}

\title{The ASEP and determinantal point processes}

\author[A. Borodin and G. Olshanski]{Alexei Borodin and Grigori Olshanski}

\date{August 4, 2016}

\begin{abstract} We introduce a family of discrete determinantal point processes related to orthogonal polynomials on the real line, with correlation kernels defined via spectral projections for the associated Jacobi matrices. For classical weights, we show how such ensembles arise as limits of various hypergeometric orthogonal polynomials ensembles. 

We then prove that the q-Laplace transform of the height function of the ASEP with step initial condition is equal to the expectation of a simple multiplicative functional on a discrete Laguerre ensemble --- a member of the new family. This allows us to obtain the large time asymptotics of the ASEP in three limit regimes: (a) for finitely many rightmost particles; (b) GUE Tracy-Widom asymptotics of the height function; (c) KPZ asymptotics of the height function for the ASEP with weak asymmetry. We also give similar results for two instances of the stochastic six vertex model in a quadrant. The proofs are based on limit transitions for the corresponding determinantal point processes. 

\end{abstract}

\maketitle

\setcounter{tocdepth}{1}
\tableofcontents

\section{Introduction}

Since the early 1960's, determinantal (and closely related Pfaffian) random point processes have served as a key tool in asymptotic analysis of exactly solvable probabilistic systems in mathematics and physics. In the late 1990's, the domain of their applicability was extended to random growth models and interacting particle systems in (1+1) dimensions, see, e.g., the surveys of Johansson \cite{J-survey}, Ferrari-Spohn \cite{FS-survey}, Borodin-Gorin \cite{BG-lectures}, and references therein. 

About ten years ago, the work of Tracy-Widom \cite{TW2, TW3} on asymptotics of the partially asymmetric simple exclusion process (ASEP, for short) started a new wave of developments. The ASEP is one of the most basic interacting particle systems whose large time asymptotics did not seem to be susceptible to standard determinantal or Pfaffian methods. Tracy-Widom employed a different approach (coordinate Bethe ansatz), and a flurry of activity followed, see, e.g., a survey of Corwin \cite{C-survey} and references therein. 

Tracy-Widom showed that the large time fluctuations of current for the ASEP with step initial data were described by the GUE Tracy-Widom distribution, which originated from the Airy determinantal point process \cite{TW1}. However, before the large time limit the determinantal processes were nowhere to be seen. 

Some hope for a greater involvement of the determinantal point processes appeared with the work on asymptotics of directed polymers in random media. For the point-to-point continuum Brownian polymer (equivalently, the Kardar-Parisi-Zhang stochastic partial differential equation with so-called narrow wedge initial data), that is known to be a limiting object for the ASEP due to Bertini-Giacomin \cite{BertiniGiacomin}, it was shown by Amir-Corwin-Quastel \cite{ACQ}, Calabrese-Le Doussal-Rosso \cite{CDR}, Dotsenko \cite{D}, and Sasamoto-Spohn \cite{SS}, that the Laplace transform of the distribution of its partition function is equal to an average of a simple multiplicative functional on the Airy determinantal process (cf. a recent note of Borodin-Gorin \cite{BG-KPZ} for this form of the result).

For another, semi-discrete Brownian polymer model known as the O'Connell-Yor polymer \cite{OCY}, it was shown that the Laplace transform of the partition function can be realized via averages over signed determinantal point processes, see O'Connell \cite{OC}, Imamura-Sasamoto \cite{IS}. Unfortunately, signed (i.e., non-positive) measures are often of limited probabilistic use, although Imamura-Sasamoto were able to take the limit of their result to see the KPZ-Airy connection of the previous paragraph. 

One goal of this work is to make an explicit connection between the ASEP and determinantal point processses, and to show how this connection can be used for analyzing the large time asymptotics. 

The ASEP can be realized as a limit of another random growth system known as the stochastic six vertex model; its definition goes back to Gwa-Spohn \cite{GwaSpohn}.
Very recently, Borodin \cite{B-6v} noticed that certain averages over the stochastic six vertex model coincide with other averages for the so-called Schur measures (introduced by Okounkov in \cite{Ok-schur}); the latter can be thought of as prototypical examples of the determinantal point processes; see, e.g., \cite{BG-lectures} for detailed explanations. 

Taking the ASEP limit of this coincidence is not entirely straightforward, and this is the first main result of the present paper. The family of determinantal processes that corresponds to the ASEP (with step initial data) turns out to be a novel one. We call them the discrete Laguerre ensembles; they live on $\Z_{\ge 0}:=\{0,1,2,\dots\}$, and their correlation kernels are  expressed through the classical Laguerre orthogonal polynomials. We prove that the q-Laplace transform of the ASEP height function is equal to the average of a multiplicative functional on the corresponding discrete Laguerre ensemble. 

We then show how this result implies three different asymptotic regimes for the ASEP. They correspond to two limit regimes of the discrete Laguerre ensemble. In the first limit, that deals with finitely many first ASEP particles, the discrete Laguerre ensemble converges to the discrete Hermite ensemble (that goes back to Borodin-Olshanski \cite{BO-JAlg}). In the second and third limits, that correspond to the ASEP height function convergence to the GUE Tracy-Widom distribution and the solution of the KPZ equation mentioned above, the discrete Laguerre ensemble converges to the Airy process. The difference between these two limits on the side of the discrete Laguerre ensemble is provided solely by different asymptotic behavior of the multiplicative functional. 

We also explain what the corresponding limits mean for the stochastic six vertex model (the convergence to the GUE Tracy-Widom distribution was previously obtained in \cite{BCG} and \cite{B-6v}). 

In another direction, we introduce the discrete Jacobi ensemble, explain the operator-theoretic mechanism of how all our discrete ensembles arise from the theory of classical orthogonal polynomials, and exhibit numerous limit transitions between more classical orthogonal polynomial ensembles and the new ones. 

Let us now explain our results in more detail. 

In what follows, we assume the reader's familiarity with the basic definitions and properties of the determinantal point processes; cf. Section \ref{sect2.1} and references therein. 

Rather than introducing the discrete Laguerre ensemble (denoted by $\DL$ below) by a formula, let us explain a general construction, of which $\DL$ is a particular case. 

Let $\W=\W(dt)$ be a measure on $\R$ such that (a) it is absolutely continuous with respect to the Lebesgue measure $dt$; (b) it has finite moments of any order; (c) the moment problem for $\W$ is determinate.
Let $\wt\PP_0, \wt\PP_1,\dots$ denote the orthonormal polynomials with respect to $\W$ with positive highest coefficients; they form a basis in $\mathcal H:=L^2(\R,\W)$. Given a point $\r\in\R$ inside the support of $\W$, consider the orthogonal decomposition $\mathcal H=\mathcal H^-_\r\oplus \mathcal H^+_\r$, where $\mathcal H^-_\r\subset\mathcal H$ and $\mathcal H^+_\r\subset\mathcal H$ are the subspaces of functions supported by $(-\infty,\r)$ and $(\r,+\infty)$, respectively.

We have an isomorphism of Hilbert spaces $\mathcal H\leftrightarrow \ell^2(\Z_{\ge0})$ by means of the correspondence $\wt\PP_n\leftrightarrow \de_n$, $n\in\Z_{\ge0}$.  Under this isomorphism, the decomposition $\mathcal H=\mathcal H^-_\r\oplus \mathcal H^+_\r$ induces an orthogonal decomposition $\ell^2(\Z_{\ge0})=L^-_\r\oplus L^+_\r$. Denote by $K^-_\r$ and $K^+_\r$ the orthogonal projections onto $L^-_\r$ and $L^+_\r$, respectively. These operators define determinantal point processes $\P^\pm_\r$ on $\Z_{\ge0}$. 
The correlation kernels of these point processes (i.e., the matrices of $K^\pm_\r$) have the form
\begin{equation*}
\begin{gathered}
K^+_\r(x,y)=\int_\r^{+\infty} \wt\PP_x(t) \wt\PP_y(t)\W(dt), \quad
K^-_\r(x,y)=\int_{-\infty}^\r \wt\PP_x(t) \wt\PP_y(t)\W(dt), \quad x,y\in\Z_{\ge0}.
\end{gathered}
\end{equation*}

Choosing $\W(dt)$ to be one of the three classical weights 
$$
\exp(-t^2)\,dt, \qquad \mathbf{1}_{t>0}\,t^{\beta-1}\exp(-t)\,dt, \qquad \mathbf{1}_{-1<t<1}\,(1-t)^a(1-t)^b \,dt, \qquad t\in \R,
$$
we arrive at the discrete Hermite, Laguerre, and Jacobi ensembles, respectively. They are very different from the orthogonal polynomial ensembles associated with these weights; those live on $\R$ and have finitely many particles almost surely; cf. Section \ref{sc:orth-poly-ens}. Note that the roles of the index and the independent variable of the orthogonal polynomials are swapped when one moves from one type of ensembles to the other, which bears certain similarity to the idea of \emph{bispectrality}, cf. Gr\"unbaum \cite{Gr}. 
 One relation between the two types of ensembles is discussed in Section \ref{sc:duality}. The discrete Hermite ensembles previously appeared in Borodin-Olshanski \cite{BO-JAlg}, two special cases of the discrete Jacobi kernel previously appeared in Borodin-Kuan \cite{BK}, the discrete Laguerre and the general discrete Jacobi ensembles appear to be new.

A key feauture of $\P^\pm_\r$ that is important to us, is that for the classical weights, the correlation kernels are spectral projections for rather simple second order difference operators on $\Z_{\ge0}$. These are
noting else but the \emph{Jacobi matrices} associated with the corresponding systems of orthogonal polynomials. In other words, they are tridiagonal $\Z_{\ge0}\times\Z_{\ge0}$ matrices that represent operators of the form $\pm(T-\const)$ in the basis of the corresponding orthogonal polynomials, where $T$ is the operator of multiplication by the independent variable $t$ in $L^2(\R,\W(dt))$. The concrete form of these matrices for the classical weights can be found in Section \ref{sect5} below.  

Let us explain why this is important to us. The classical hypergeometric polynomials of the Askey scheme 
can always be viewed as eigenfunctions of a suitable second order differential or difference operator (with polynomial coefficients). Hence, the correlation kernels of the corresponding orthogonal polynomial ensembles can be viewed as spectral projections for the associated operators.

 There are several asymptotic regimes, in which the orthogonal polynomials ensembles cease to be such (for example, the number of particles may tend to infinity), but the associated differential/difference operators are easily seen to have a limit. In the case when the state spaces before and after the limit are discrete, this is actually sufficient to claim the convergence of the spectral projections for such operators, thus the convergence of the corresponding determinantal point processes. Moreover, the computations are significantly simpler than in more traditional approaches to such asymptotics that typically involve steepest descent arguments or the Riemann-Hilbert problem methods. However, if some of the involved state spaces are continuous, we are currently not able to produce a rigorous justification of this method, and it should be viewed as heuristic. However, in all the examples we know, it does yield a correct statement with rather modest computations, and it also provides possibly the simplest way of guessing the correct scaling.  

The difference operator approach to asymptotics of the correlation kernels was first used by Borodin-Olshanski in \cite{BO-JAlg}. It was also applied by Borodin-Gorin \cite{BG}, Gorin \cite{G}, Olshanski \cite{Ol-FAA}, and (independently) discussed in the continuous setup by Tao \cite{Tao-blog, Tao-book}; see also  Bornemann \cite{Bornemann}.  
It is also somewhat similar in spirit to the idea of Edelman-Sutton \cite{ES} that lead to spectacular progress in
understanding limits of the general $\beta$-ensembles in Random Matrix Theory, with the key difference that
one has to deal with \emph{random} tridiagonal matrices in that case (which is much harder). Another related idea can be found in the work of Breuer--Duits \cite{BD-JAMS, BD-CMP}, who used the asymptotics of Jacobi matrices to prove Gaussian fluctuations for the corresponding orthogonal polynomial ensembles with growing number of particles. 

In Section \ref{sect6} below, we use our approach to prove several limiting statements, realizing the
discrete Hermite, Laguerre, and Jacobi ensembles as limits in different ways, cf. Figure \ref{fig} in that section. One of those statements
deals with convergence of the Meixner orthogonal polynomial ensembles to the discrete Laguerre ensemble.
It is this limit that, coupled with a result of \cite{B-6v} on matching observables of the stochastic
six vertex model and the Schur measures (of which Meixner ensembles are a special case), leads to our first main result:

\begin{theorem*} Consider the ASEP on $\Z$ with particles occupying all negative integers at time $0$, and with the left jump rate $\l=q\in (0,1)$ and right jump rate $\mathfrak r=1$. Let $\HT(x)$ denote the number of the ASEP particles weakly to the right of the position $x\in\Z$. Then at any time moment $t\ge 0$, and for any $x\ge 0$, $\zeta\in\C\setminus\{-q^{\Z_{\le 0}}\}$,  we have
\begin{equation*}
\E_{ASEP}\prod_{i\ge 1} \frac 1{1+\zeta q^{\HT(x)+i}}= \E_{Z\in\DL^+((1-q)t,x+1)}\prod_{z\in Z} \frac 1{1+\zeta q^z},
\end{equation*}
where $\DL^+(r;\beta)$ is the determinantal point process of the form $\P^+_r$ with the Laguerre weight $\W(dt)=\mathbf{1}_{t>0}\,t^{\beta-1}\exp(-t)\,dt$. A similar relation holds for $x<0$ as well.
\end{theorem*} 

A further limit transition from the discrete Laguerre ensemble to the discrete Hermite one, which we also prove with the operator method, leads to the asymptotics of the first ASEP particles at large times.
 
Another limit transition, from the discrete Laguerre ensemble to the Airy one, involves a continuous state space (for Airy), so we provide heuristics with our approach, which in particular gives us the (not so trivial) correct normalization for the limit. The only piece of this limit transition that is needed for
the GUE Tracy-Widom asymptotics proof for the ASEP is the convergence of the distribution of the first particle, and that is actually equivalent (thanks to the results of Section \ref{sc:duality}) to the convergence of the first particle of the Laguerre orthogonal polynomial ensemble to the first Airy particle --- a well-known fact that goes back to \cite{Joh-shape}. On the other hand, the convergence of the weakly asymmetric ASEP to KPZ requires the convergence of the full discrete Laguerre ensemble to Airy, and we sketch the argument, assuming that convergence. A rigorous proof would require a verification of the trace-class convergence of the kernels by a different method (e.g., by classical steepest descent arguments), which is standard but technical, and we do not pursue it.

Finally, we utilize the convergences of the Meixner and Charlier orthogonal polynomial ensembles to the discrete Hermite and Airy ensembles to give similar large-scale asymptotic statements for two instances
of the (higher spin) stochastic six vertex model in a quadrant. 

\medskip

\noindent\textbf{Acknowldegements.}\ We are very grateful to Vadim Gorin for very helpful remarks. The work of A.~B.~ was partially supported by the NSF grants DMS-1056390 and DMS-1607901.

\section{Determinantal measures and orthogonal polynomials}\label{sect2}

\subsection{Generalities}\label{sect2.1}
Let $\X$ be a locally compact Polish space. By a  \emph{point configuration} in $\X$ we mean an arbitrary subset $X\subset\X$ without accumulation points; it is either finite or countable. With any $X\in\Conf(\X)$ one associates the atomic measure $\de_X:=\sum_{x\in X}\de_x$ (the sum of delta-measures at the points of $X$). 

The space of all point configurations is denoted by $\Conf(\X)$; it can be endowed with a natural structure of measurable space. Given a probability measure $\P$ on $\Conf(\X)$, one can speak of a \emph{random} point configuration $X$. Likewise, $\de_X$ becomes a random Radon measure on $\X$. Averaging $\de_X$ with respect to $\P$ one obtains a (non-random) measure $\rho_1$ on $\X$, called the \emph{density measure} of $\P$. It is the simplest invariant of $\P$. One can generalize this construction and obtain an infinite sequence of invariants  $\rho_1,\rho_2,\dots$, where $\rho_k$ is a symmetric measure on $\X^k$ (the $k$-fold product $\X\times\dots\times\X$) called the \emph{$k$-point correlation measure}. 

$\P$ is said to be a \emph{determinantal measure}, or a \emph{determinantal point process},  if one can exhibit  a Radon measure $\nu$ on $\X$ (called the \emph{reference measure}) and a function $K(x,y)$ on $\X\times\X$ (called the \emph{correlation kernel}) such that for each $k=1,2,\dots$, $\rho_k$ has a density $\rho_k(x_1,\dots,x_k)$ with respect to $\nu^{\otimes k}$, and this density (called the \emph{$k$-point correlation function}) is given by symmetric minors of the kernel:
$$
\rho_k(x_1,\dots,x_k)=\det[K(x_i,x_j)]_{i,j=1}^k.
$$

If such a pair $\{\nu$, $K(x,y)\}$ exists, then it determines $\P$ uniquely. On the other hand, different pairs may produce the same determinantal measure. For instance, one can replace $\nu$ by an equivalent measure $f\nu$ and at the same time replace the kernel $K(x,y)$ by the new kernel $K(x,y)f^{-1/2}(x)f^{-1/2}(y)$; then the correlation measures do not  change. Another possibility is to keep $\nu$ fixed but replace $K(x,y)$ with $K(x,y)\epsi(x)\epsi^{-1}(y)$, where $\epsi(x)$ is a nonvanishing function on $\X$; this transformation also does not affect the correlation measures. 

Which pairs $\{\nu$, $K(x,y)\}$ give rise to determinantal measures is a difficult question if one does not impose  additional assumptions. Usually one wants $\nu$ to be some natural measure (say, the Lebesgue measure on $\R^n$, or the counting measure when $\X$ is a discrete space), but the problem is related to a description of admissible
kernels $K(x,y)$. 

However, in the present paper we are dealing with a very special class of correlation kernels, and for our purposes the following abstract existence theorem is sufficient. 

Let $L$ be a closed subspace of the Hilbert space $L^2(\X,\nu)$, and $K$ be the orthogonal projection onto $L$. The kernel $K(x,y)$ of the operator $K$ is the \emph{reproducing kernel} of $L$, which is defined by
$$
K(x,y)=\sum_n f_n(x)\overline{f_n(y)}, \qquad x,y\in\X,
$$
where $\{f_n\}$ is an arbitrary orthonormal basis of $L$; the definition does not depend on the choice of the basis. Assume that the function $x\mapsto K(x,x)$ is locally $\nu$-integrable; then there exists a (unique) determinantal measure $\P=\P_K$ for which $K(x,y)$ serves as a correlation kernel. For a more accurate formulation of this result, see Soshnikov \cite{So}.

Note that one can change $L$ and hence $K$ without changing the corresponding measure $\P_K$. Namely, let $\epsi(x)$ be a function on $\X$ with the values on the unit circle in $\C$, and let $\epsi L$ consist of the functions of the form $\epsi(x)f(x)$, where $f$ ranges over $L$. Then replacing $L$ with $\epsi L$ does not affect the determinantal measure. Even if one wants $\epsi(x)$ to be real-valued, it may take values $\pm1$. Concrete examples will be given below. 

If $L\subset L^2(\X,\nu)$, the range of $K$, has finite dimension, then the measure $\P_K$ always exists, and it is concentrated on the subspace $\Conf_N(\X)$ of $N$-point configurations, where $N=\dim L$ (the converse is also true). If $L$ has infinite dimension and $\P_K$ exists, then the $\P_K$-random configuration $X$ is infinite, with probability 1.  

\subsection{The discrete case}\label{sect2.2}

In this section we assume that $\X$ is a finite set or a countable set with discrete topology (this is what we mean by the discrete case). Then $\Conf(\X)$ is simply the set $2^\X$ of all subsets of $\X$. If $\X$ is finite, then $\Conf(\X)$ is a finite set, too. If $\X$ is countable, then $\Conf(\X)$ is a compact, totally disconnected space with respect to the natural topology --- the base of the topology is formed by the cylinder sets 
$$
C_Y:=\{X\in\Conf(\X): \textrm{$X$ contains a given finite set $Y\subset\X$}\}.
$$

We are mainly interested in the case when $\X$ is countable, but occasionally we will need finite sets $\X$ as well.

We take as $\nu$ the counting measure on $\X$.  The correlation functions of a probability measure $\P$ on $\Conf(\X)$ admit a simple interpretation: for a $k$-tuple of distinct points $Y=\{x_1,\dots,x_k\}$, the value $\rho_k(x_1,\dots,x_k)$  is  equal to $\P(C_Y)$; in other words, $\rho_k(x_1,\dots,x_k)$ is the probability that the random configuration contains all the points $x_1,\dots,x_k$. 

For the counting measure $\nu$, the Hilbert space $L^2(\X,\nu)$ turns into the coordinate space $E:=\ell^2(\X)$ with its distinguished basis $\{e_x\}$ indexed by the points of $\X$. In the discrete case,  the subtleties related to an accurate definition of a reproducing kernel disappear, and any closed subspace $L\subset E$ gives rise to a determinantal measure.

In the discrete case, there exists a special operation called the \emph{particle/hole involution}. This is the involutive self-map $\Conf(\X)\to\Conf(\X)$ assigning to a subset $X$ its complement $X^\circ:=\X\setminus X$. The correspondence $X\leftrightarrow X^\circ$ induces, in a natural way,  an involutive map $\P\mapsto\P^\circ$ on the set of probability measures on $\Conf(\X)$. On the subset of determinantal measures, the latter map takes the form $(\P_K)^\circ=\P_{1-K}$.

\subsection{Orthogonal polynomial ensembles}\label{sc:orth-poly-ens}
Here we define a class of determinantal point processes associated with orthogonal polynomials. Consider a system $P_0,P_1,P_2,\dots$ of orthogonal polynomials with a weight measure $W$ on $\R$, and let $\supp W$ be the support of $W$. The system $\{P_n\}$ may be fairly general, but, to slightly simplify things, let us assume that it is taken from the Askey scheme (see Koekoek--Swarttouw \cite[Chapter 1]{KS}). Then $\supp W\subset\R$ is either a discrete subset or a closed interval (possibly, with infinite ends). In the latter case we denote by $W(x)$ the density of $W$ with respect to the Lebesgue measure.)

We set $\X:=\supp W$ and $\nu:=W$, and for $N=1,2,\dots$ we denote by $L_N$ the $N$-dimensional subspace of $L^2(\X,W)$ formed by the polynomial functions of degree $\le N-1$ (if the set $\X$ is finite, then we suppose that $N$ is smaller that its size). Let $K_N$ be the projection onto $L_N$, and $\P_{K_N}$ be the corresponding determinantal measure. The probability space $(\Conf_N(\X),\P_{K_N})$ is called the \emph{$N$-particle orthogonal polynomial ensemble}, see Koenig \cite{Koenig} for a survey.   

Its correlation kernel, taken with respect to the reference measure $\nu=W$, is the \emph{Christoffel--Darboux kernel}
$$
\sum_{n=0}^{N-1}\frac{P_n(x)P_n(y)}{\Vert P_n\Vert^2}, \qquad x,y\in\X,
$$
where the norm is that of the weighted Hilbert space $L^2(\X,W)$.  As is well known, the Christoffel--Darboux kernel can also be written in the form
$$
\frac{k_{N-1}}{k_N\Vert P_{N-1}\Vert^2}\frac{P_N(x)P_{N-1}(y)-P_{N-1}(x)P_N(y)}{x-y},
$$
where $k_n$ denotes the leading coefficient in $P_n$.

If instead we take as $\nu$ the counting measure (in the discrete case) or the Lebesgue measure (in the continuous case), then the kernel should be multiplied by the factor $(W(x)W(y))^{\frac12}$; we write it as 
$$
K_N(x,y)=(W(x)W(y))^{\frac12}\sum_{n=0}^{N-1}\frac{P_n(x)P_n(y)}{\Vert P_n\Vert^2}, \qquad x,y\in\X.
$$

\section{Discrete Hermite, Laguerre, and Jacobi ensembles}\label{sect3}

\subsection{Discrete ensembles associated with continuous orthogonal polynomials}\label{sect3.1}

Let $\W=\W(dt)$ be an arbitrary measure on $\R$ with the following properties:

(i) $\W$ is absolutely continuous with respect to the Lebesgue measure $dt$; 

(ii) $\W$ has finite moments of any order;

(iii) the moment problem for $\W$ is determinate.

Condition (i) is not strictly necessary but it simplifies things. Condition (ii) ensures the existence of an infinite system $\PP_0, \PP_1,\dots$ of orthogonal polynomials; we denote the corresponding orthonormal system by $\wt{\PP}_0, \wt{\PP}_1,\dots$\,. Condition (iii) implies that the space $\C[t]$ of polynomials is dense in the weighted Hilbert space $\mathcal H:=L^2(\R,\W)$, so that $\{\wt{\PP}_0, \wt{\PP}_1,\dots\}$ is a basis in $\mathcal H$ (see Akhiezer \cite[Corollary 2.3.3]{A}).

In what follows, we assume that for each $n=0,1,\dots$,  the leading coefficient of $\PP_n$ is strictly positive. This means that the polynomials $\PP_n$ are defined up to positive numeric factors, while the polynomials $\wt\PP_n=\PP_n/\Vert\PP_n\Vert^{-1}$ are uniquely defined by the weight measure. 

Given a point $\r\in\R$, we consider the orthogonal decomposition $\mathcal H=\mathcal H^-_\r\oplus \mathcal H^+_\r$, where $\mathcal H^-_\r\subset\mathcal H$ and $\mathcal H^+_\r\subset\mathcal H$ are the subspaces of functions supported by $(-\infty,\r)$ and $(\r,+\infty)$, respectively. Next, we assume that $\r$ is inside the support of $\W$, so that both these semi-infinite intervals have strictly positive mass relative to $\W$. This means that both $\mathcal H^-_\r$ and $\mathcal H^+_\r$ have infinite dimension. 

Finally, we define an isomorphism of Hilbert spaces $\mathcal H\leftrightarrow \ell^2(\Z_{\ge0})$ by means of the correspondence $\wt P_n\leftrightarrow e_n$, $n\in\Z_{\ge0}$.  Under this isomorphism, the decomposition $\mathcal H=\mathcal H^-_\r\oplus \mathcal H^+_\r$ induces an orthogonal decomposition $\ell^2(\Z_{\ge0})=L^-_\r\oplus L^+_\r$. Then we denote by $K^-_\r$ and $K^+_\r$ the projections onto $L^-_\r$ and $L^+_\r$, respectively. These projection operators are the objects of interest for us. 

\begin{definition}\label{def3.A}
Let $\W$ and $\{\PP_n\}$ be as above, $\r$ be a point inside the support of $\W$, and $K^\pm_\r$ be the corresponding self-adjoint projection operators on $\ell^2(\Z_{\ge0})$ as defined above. By the general theory (Section \ref{sect2.1}), $K^\pm_\r$ gives rise to a determinantal measure $\P^\pm_\r$ on $\Conf(\Z_{\ge0})$ (in other words, a determinantal point process on $\Z_{\ge0}$); we call $\P^\pm_\r$ the \emph{discrete ensemble associated with the system $\{\PP_n\}$}.
\end{definition}

Note that $(\P^\pm_\r)^\circ=\P^\mp_\r$. Note also that the $\P^\pm_\r$-random configuration contains almost surely infinitely many points: as pointed out in Section \ref{sect2.1}, this follows from the fact that  $K^\pm_\r$ has infinite rank.

\begin{definition}\label{def3.B}
Let $\r$ be a point inside the support of $\W$. Introduce kernels $K^\pm_\r(x,y)$ on $\Z_{\ge0}\times\Z_{\ge0}$  by setting
\begin{equation}\label{eq3.A}
\begin{gathered}
K^+_\r(x,y)=\int_\r^{+\infty} \wt\PP_x(t) \wt\PP_y(t)\W(dt), \quad x,y\in\Z_{\ge0},\\
K^-_\r(x,y)=\int_{-\infty}^\r \wt\PP_x(t) \wt\PP_y(t)\W(dt), \quad x,y\in\Z_{\ge0}.
\end{gathered}
\end{equation}
We call $K^\pm_\r(x,y)$ the \emph{discrete kernel associated with the system $\{\PP_n\}$}.
\end{definition}

Obviously, $K^\pm_\r(x,y)$ is the kernel of the projection $K^\pm_\r$. Hence, $K^\pm_\r(x,y)$ serves as a correlation kernel for $\P^\pm_\r$. 

Observe that conditions (i)--(iii) hold true for the Hermite, Laguerre, and Jacobi orthogonal polynomials.  Indeed, conditions (i) and (ii) are obvious. As for condition (iii), it is obvious in the Jacobi case (because then the weight measure has bounded support), and it is well known in the Hermite and Laguerre cases (the moments of the weight measure do not grow too fast).  

Therefore, the general construction described above is applicable to these three systems of polynomials. This leads us to three families of determinantal point processes on $\Z_{\ge0}$, which we call the \emph{discrete Hermite ensemble}, \emph{discrete Laguerre ensemble}, and \emph{discrete Jacobi ensemble}. The same names are used for the corresponding kernels given by \eqref{eq3.A}. Each ensemble depends on the additional continuous parameter $\r$ and has two variants, ``plus'' and ``minus'', corresponding to the intervals $(\r,+\infty)$ and $(-\infty,\r)$.

Below we examine these kernels in more detail. In particular, we explain how to  write them in the  \emph{integrable form} (see \cite{IIKS, Deift}) by making use of the forward and backward shift operators related to the corresponding systems of orthogonal polynomials. All necessary formulas can be found in  Koekoek--Swarttouw \cite{KS}.

\subsection{The discrete Hermite ensemble}\label{sect3.2}
The two variants of this ensemble are denoted by $\DH^+(\r)$ and $\DH^-(\r)$. They correspond to the intervals $(\r,+\infty)$ and $(-\infty,\r)$, respectively. Here parameter $\r$ ranges over $
\R$.

For the \emph{Hermite polynomials} we use the standardization and notation of \cite[\S1.13]{KS}, which are the most common ones. The weight measure of the Hermite polynomials is $\W(dt)=e^{-t^2}dt$, where $t\in\R$. The $n$th Hermite polynomial, denoted as  $H_n(t)$, is specified by the property that its leading coefficient equals $2^n$.  In this standardization, 
$$
\Vert H_n\Vert^2=\pi^{1/2} 2^n n!.
$$
Using this formula we write down the integral representation of the discrete Hermite kernel (below $x,y\in\Z_{\ge0}$):
\begin{gather*}
K_{\DH^+(\r)}(x,y)=(\pi 2^{x+y}x!y!)^{-1/2}\int_\r^{+\infty} H_x(t)H_y(t)e^{-t^2}dt,\\
K_{\DH^-(\r)}(x,y)=(\pi 2^{x+y}x!y!)^{-1/2}\int_{-\infty}^\r H_x(t)H_y(t)e^{-t^2}dt.
\end{gather*}

Since $H_n(-t)=(-1)^n H_n(t)$, we have
$$
K_{\DH^-(\r)}(x,y)= (-1)^{x+y} K_{\DH^+(-\r)}(x,y).
$$
This yields the symmetry relation
$$
\DH^-(\r)=\DH^+(-\r),
$$
because the factor $(-1)^{x+y}$ does not affect the determinantal measure. 

\begin{proposition}\label{prop3.A}
For $x\ne y$, the discrete Hermite kernel can be written in the form
$$
K_{\DH^\pm(\r)}(x,y)=\mp(\pi x!y!2^{x+y+2})^{-1/2}\,e^{-\r^2} \frac{H_{x+1}(\r)
H_y(\r) -H_x(\r) H_{y+1}(\r)}{x-y}.
$$
\end{proposition}

\begin{proof}
We write
$$
(x-y)\int_\r^{+\infty}H_x(t)H_y(t)e^{-t^2}dt
=\int_\r^{+\infty}xH_x(t)H_y(t)e^{-t^2}dt-\int_\r^{+\infty}H_x(t)yH_y(t)e^{-t^2}dt.
$$
Next, we integrate by parts in the first integral using the formulas 
$$
x H_x(t)=\frac12 \frac{d}{dt} H_{x+1}(t), \qquad \frac{d}{dt}\left(H_y(t)e^{-t^2}\right)=-H_{y+1}(t)e^{-t^2},
$$
which are obtained from \cite[(1.13.6) and (1.13.8)]{KS} (the forward and backward shifts). Then we do the same with the second integral. The resulting two integral terms are cancelled out,  and  we obtain the desired formula for the kernel $K_{\DH^+(\r)}$. The case of $K_{\DH^+(\r)}$ is handled in exactly the same way. 

Note that the sign $\mp$ in the right-hand side agrees with the fact that
$$
K_{\DH^-(\r)}(x,y)+K_{\DH^+(\r)}(x,y)=\begin{cases} 1, & x=y,\\
0, & x\ne y. \end{cases}
$$
\end{proof}

\subsection{The discrete Laguerre ensemble}\label{sect3.3}
We define the \emph{Laguerre polynomials} as the orthogonal polynomials on $[0,+\infty)$ with the weight measures $t^{\be-1} e^{-t}dt$, where $\be>0$. We denote by $L^{(\be)}_n(t)$ the $n$th Laguerre polynomial; in our standardization, its leading coefficient equals $1/n!$. This slightly differs from the conventional definition: the connection with the notation of \cite[\S1.11]{KS} is the following: 
$$
L^{(\be)}_n:=(-1)^n L^{\be-1}_n(t), 
$$
where $L^\al_n$ is the $n$th Laguerre polynomial in the standardization of  \cite[\S1.11]{KS}. 

Here is the formula for the norm:
$$
\Vert L^{(\be)}_n\Vert^2=\frac{\Ga(n+\be)}{n!}. 
$$

The two variants of the discrete Laguerre ensemble are denoted by $\DL^\pm(\r;\be)$. The corresponding discrete Laguerre kernel has the form
\begin{gather*}
K_{\DL^+(\r;\be)}(x,y)=\left(\frac{x!y!}{\Ga(x+\be)\Ga(y+\be)}\right)^{1/2}\int_\r^{+\infty} L^{(\be)}_x(t)L^{(\be)}_y(t)t^{\be-1}e^{-t}dt,\\
K_{\DL^-(\r;\be)}(x,y)=\left(\frac{x!y!}{\Ga(x+\be)\Ga(y+\be)}\right)^{1/2}\int_{-\infty}^\r L^{(\be)}_x(t)L^{(\be)}_y(t)t^{\be-1}e^{-t}dt.
\end{gather*}

In contrast to the case of the Hermite polynomials (and that of the Jacobi polynomials, see below), there is no symmetry relation which would reduce the first integral to the second one.  

\begin{proposition}\label{prop3.B}
For $x\ne y$, the discrete Laguerre kernel can be written in the form
\begin{multline*}
K_{\DL^\pm(\r;\be)}(x,y)=\pm\left(\frac{x!y!}{\Ga(x+\be)\Ga(y+\be)}\right)^{1/2}\r^\be e^{-\r}\cdot
\frac{L^{(\be+1)}_{x-1}(\r)L^{(\be)}_y(\r)-L^{(\be)}_x(\r)L^{(\be+1)}_{y-1}(\r)}{x-y}
\end{multline*}
with the convention that $L^{(\be)}_{-1}(\r)=0$.
\end{proposition}

\begin{proof}
As in the proof of Proposition \ref{prop3.A}, we apply the same trick with integration by parts. But now we use the backward shift first and the forward shift next. More precisely, these are the formulas
$$
xL^{(\be)}_x(t)t^{\be-1}e^{-t}=\frac{d}{dt}\left(-L^{(\be+1)}_{x-1}(t)t^\be e^{-t}\right), \quad \frac{d}{dt}L^{(\be)}_y(t)=L^{(\be+1)}_{y-1}(t),
$$
which are derived from \cite[(1.11.8) and (1.11.6)]{KS}.
\end{proof}

\subsection{The discrete Jacobi ensemble}\label{sect3.4}
The Jacobi polynomials are defined as in \cite[\S1.8]{KS}. They depend on two parameters $a>-1$, $b>-1$, and are denoted by $P_n^\ab(t)$.  The argument $t$ and the additional parameter $\r$ range over $(-1,1)$. The two variants of the discrete Jacobi ensembles are denoted by $\DJ^\pm(\r;a,b)$.   

Using the formula
\begin{equation*}
\W(dt)=(1-t)^a(1+t)^bdt, \quad -1<t<1,
\end{equation*}
we write the integral representation \eqref{eq3.A} as
\begin{gather*}
K_{\DJ^+(\r;a,b)}=\frac1{\Vert P_x^\ab\Vert \Vert P_y^\ab\Vert}\int_\r^1 P^\ab_x(t) P^\ab_y(t)(1-t)^a(1+t)^bdt, \\
K_{\DJ^-(\r;a,b)}=\frac1{\Vert P_x^\ab\Vert \Vert P_y^\ab\Vert}\int_{-1}^\r P^\ab_x(t) P^\ab_y(t)(1-t)^a(1+t)^bdt,
\end{gather*}
where the explicit expression for the norm is 
$$
\Vert P^\ab_n\Vert^2=\frac{2^{a+b+1}\Ga(n+a+1)\Ga(n+b+1)}{(2n+a+b+1)\Ga(n+a+b+1)n!}.
$$

As in the case of the discrete Hermite ensemble, we have a symmetry relation; now it takes the form (note the $a\leftrightarrow b$ swap)
\begin{gather*}
K_{\DJ^-(\r;a,b)}(x,y)=(-1)^{x+y}K_{\DJ^+(-\r;b,a)}(x,y),\\
\DJ^-(\r;a,b)=\DJ^+(-\r;b,a).
\end{gather*}

\begin{proposition}\label{prop3.C}
For $x\ne y$, the discrete Jacobi kernel can be written in the form
\begin{multline*}
K_{\DJ^\pm(\r;a,b)}(x,y)=\pm\frac{(1-\r)^{a+1}(1+\r)^{b+1}}{2\Vert P^\ab_x\Vert \Vert P^\ab_y\Vert}\\
\times \frac{(x+a+b+1)P^{a+1,b+1)}_{x-1}(\r) P^\ab_y(\r)-P^\ab_x(\r)(y+a+b+1)P^{(a+1,b+1)}_{y-1}(\r)}{\wt x -\wt y},
\end{multline*}
where
$$
\wt x:=x(x+a+b+1), \quad \wt y:=y(y+a+b+1), \quad P^{(a+1,b+1)}_{-1}(\r):=0.
$$
\end{proposition}

\begin{proof}
The same trick as in Propositions \ref{prop3.A} and \ref{prop3.B}. In the case of $\DJ^+(\r;a,b)$ we write
\begin{multline*}
(\wt x-\wt y)\int_\r^1 P^\ab_x(t) P^\ab_y(t)(1-t)^a(1+t)^bdt\\
=(x+a+b+1)\int_\r^1 \left\{x P^\ab_x(t)(1-t)^a(1+t)^b\right\} P^\ab_y(t)dt\\
-(y+a+b+1)\int_\r^1 P^\ab_x(t)\left\{y P^\ab_y(t)(1-t)^a(1+t)^b\right\}dt.
\end{multline*}
Next, we integrate by parts in the first integral using the formulas
$$
x P^\ab_x(t)(1-t)^a(1+t)^b=-\frac12\frac{d}{dt}\left(P^{(a+1,b+1)}_{x-1}(t)(1-t)^{a+1}(1+t)^{b+1}\right)
$$
and
$$
\frac{d}{dt} P^\ab_y(t)=\frac{y+a+b+1}2 \,P^{(a+1,b+1)}_{y-1}(t),
$$
which are obtained from \cite[(1.8.8)]{KS} (backward shift) and \cite[(1.8.6)]{KS} (forward shift), respectively. Then we do the same with the second integral. Again, the resulting integral terms are cancelled out, and we obtain the desired expression for the kernel $K_{\DJ^+(\r;a,b)}(x,y)$. The case of $K_{\DJ^-(\r;a,b)}(x,y)$ is handled in exactly the same way.
\end{proof}

\begin{remark}
The fact that we obtain (in the denominator) the difference $\wt x- \wt y$ instead of $x-y$, is caused by the very structure of the forward and backward shift operators in the Jacobi case. Trying $x-y$, as before, we could not kill the integral terms. The transformation $x\mapsto x(x+a+b+1)$ of the discrete index $x$ agrees with the limit transition discussed in Section \ref{Racah-DJacobi} below. 
\end{remark}

\subsection{Duality between continuous and discrete ensembles} \label{sc:duality}

Let us return to the general setting of Section \ref{sect3.1} and consider the discrete ensembles $\P^\pm_\r$ linked to a system $\{\wt\PP_x: x\in\Z_{\ge0}\}$ of orthonormal polynomials with a weight measure $\W(dt)$ satisfying the three conditions (i)--(iii). To handle the two variants together, let us use the common notation $\P_I:=\P^\pm_\r$, where $I$ denotes either the interval $(\r,+\infty)$ (in the case of $\P^+_\r$) or the interval $(-\infty,r)$ (in the case of $\P^-_\r$).

Next,  let us fix a natural number $N$ and denote by $\wh\P_N$ the continuous $N$-particle orthogonal polynomial ensemble coming from the same system of polynomials. We are going to show that certain two gap probabilities, which are related to $\P_I$ and $\wh\P_N$, respectively, are the same.

\begin{theorem}\label{thm3.A}
Let $X$ denote the\/ $\P_I$-random infinite configuration on $\Z_{\ge0}$, and $Y$  denote the\/ $\wh\P_N$-random $N$-particle configuration on $\R$.  Then the following gap probabilities are the same{\rm:}
$$
\Prob\left(X\cap\{0,1,\dots,N-1\}=\varnothing\right)=\Prob\left(Y\cap I=\varnothing\right).
$$
\end{theorem}

Note that the probabilities in the left-hand side determine the distribution of the leftmost particle in $\P^\pm_\r$, while the probabilities in the right-hand side do the same for the leftmost or the rightmost particle in $\wh\P_N$, depending on whether $I=(\r,+\infty)$ or $I=(-\infty,\r)$. 

As is seen from the proof, the special form of the interval $I$ is inessential here: it could be replaced an arbitrary Borel subset of positive $\W$-mass.

\begin{proof}
We apply the well known fact that gap probabilities for determinantal processes are given by Fredholm determinants (see, e.g., \cite[Theorem 2]{So}, where one has to specialize $z=0$). Then the equality in question reduces to the following one: 
\begin{equation}\label{eq3.B}
\det(1-K_{I,N})=\det(1-\wh K_{N,I}),
\end{equation}
where $K_{I,N}$ is the $N\times N$ matrix with the entries
$$
K_{I,N}(x,y)=\int_{t\in I}\wt\PP_x(t)\wt\PP_y(t)\W(dt), \qquad x,\, y\in\{0,1,\dots,N-1\},
$$
and $\wh K_{N,I}$ is the rank $N$ operator on the Hilbert space $L^2(I,\W(dt))$ given by the kernel
$$
\wh K_{N,I}(s,t)=\sum_{x=0}^{N-1}\wt\PP_x(s)\wt\PP_x(t), \qquad s,\, t\in I.
$$

Comparing these two expressions we see that the desired equality \eqref{eq3.B} is a continual analogue of the identity 
$$
\det(1-AB^*)=\det(1-B^*A)
$$
for two rectangular matrices $A,B$ of the same format. 

The proof of \eqref{eq3.B} is a routine exercise. Indeed, the left-hand side is equals to
$$
1+\sum_{m=1}^N\frac{(-1)^m}{m!}\sum_{x_1=0}^{N-1}\dots\sum_{x_m=0}^{N-1}\det[K_{I,N}(x_i,x_j)]_{i,j=1}^m,
$$
while the right-hand side is equal to
$$
1+\sum_{m=1}^{N}\frac{(-1)^m}{m!}\int_{t_1\in I}\dots\int_{t_m\in I}\det[\wh K_{N,I}(t_i,t_j)]_{i,j=1}^m \W(dt_1)\dots \W(dt_m).
$$

Then we have to show that for each fixed $m$, the $m$-fold sum over $x_1,\dots,x_m$ is equal to the $m$-fold integral over $t_1,\dots,t_m$. This is verified directly, by expanding each minor into a sum over permutations and  making use of the definition of $K_{I,N}$ and $\wh K_{N,I}$. 
\end{proof}

\section{Limit transitions: outline of the method}\label{sect4}

Here we describe, in general form, our approach to studying limit transitions for certain determinantal point processes. Concrete examples are given in the subsequent sections. 

\subsection{Generalities}

We are dealing with determinantal measures of the form $\P_K$, where $K$ is a projection operator acting on $\ell^2(\X)$ and $\X$ is a countable set (see Section \ref{sect2} above).   

\begin{proposition}\label{prop4.A}
Let $K_1, K_2,\dots$ be an infinite sequence of projection operators acting on $\ell^2(\X)$, and $K$ be one more projection. If the kernels  $K_N(x,y)$ converge to the kernel $K(x,y)$ pointwise on $\X\times\X$, then the determinantal measures $\P_{K_N}$ weakly converge to the determinantal measure $\P_K$.
\end{proposition}

\begin{proof}
From the interpretation of the correlation functions given in Section \ref{sect2.2} one sees that if $K_N(x,y)\to K(x,y)$ on $\X\times\X$, then $\P_{K_N}(C_Y)\to \P_K(C_Y)$ for an arbitrary cylinder set $C_Y\subset\Conf(\X)$. By the definition of the topology in $\Conf(\X)$,  this implies the weak convergence $\P_{K_N}\to\P_K$.
\end{proof}

Note that the pointwise convergence $K_N(x,y)\to K(x,y)$ is equivalent to  the weak convergence of projection operators $K_N\to K$, which in turn is equivalent to their strong convergence, because on the set of projections, the weak and strong operator topologies coincide. 

Recall (see \cite[\S VIII.7]{RS}) that a sequence $A_1,A_2,\dots$ of self-adjoint operators on a Hilbert space converges to a self-adjoint operator $A$ in \emph{strong resolvent sense} if their resolvents converge strongly:
$$
\textrm{$(\la-A_N)^{-1}\to(\la-A)^{-1} $ strongly, for every $\la\in\C\setminus\R$}.
$$

Given a self-adjoint operator $A$ and an open interval $(\r_1,\r_2)\subset\R$, possibly semi-infinite, we will denote by $[A]_{(\r_1,\r_2)}$ the spectral projection of $A$ corresponding to $(\r_1,\r_2)$. We also abbreviate $[A]_+:=[A]_{(0,+\infty)}$.

\begin{proposition}\label{prop4.B}
Let $A_N$, $N=1,2,\dots$, and $A$ be self-adjoint operators on a Hilbert space and suppose that $A_N\to A$ in the strong resolvent sense.  Next, let $(\r_1,\r_2)\subset\R$ be an arbitrary open interval, possibly semi-infinite and such that its finite ends are not in the point spectrum of $A$. Then the spectral projections $[A_N]_{(\r_1,\r_2)}$ strongly converge to the spectral projection $[A]_{(\r_1,\r_2)}$.
\end{proposition}

\begin{proof}
For finite intervals, this assertion is contained in \cite[{Theorem VIII.24, claim (b)}]{RS}. In the case of a semi-infinite interval the argument is the same. 
\end{proof}

Recall one more general definition. Let $A$ be a closed operator on a Banach space (in particular, on a Hilbert space). A dense subspace $\mathcal L$ contained in the domain of $A$ is said to be a \emph{core} of $A$ if the closure of the operator $A\big|_{\mathcal L}$ (the restriction of $A$ to $\mathcal L$) coincides with $A$. In particular, if $A$ is a self-adjoint operator on a Hilbert space and $\mathcal L$ is a core of $A$, then the operator $A\big|_{\mathcal L}$ is essentially self-adjoint. 

The next result provides an effective tool for checking the strong resolvent convergence of self-adjoint operators.

\begin{proposition}\label{prop4.C}
Let $A_N$, $N=1,2,\dots$, and $A$ be self-adjoint operators on a Hilbert space $\mathcal H$, and suppose that there exists a dense subspace $\mathcal L\subset \mathcal H$ such that $\mathcal L$  is a common core for all these operators and  $A_N v\to A v$ as $N\to\infty$, for any vector $v\in \mathcal L$. Then $A_N\to A$ in the strong resolvent sense.
\end{proposition}

\begin{proof}
See \cite[{Theorem VIII.25}]{RS}.
\end{proof}

\subsection{Tridiagonal and Jacobi matrices}

A matrix $\A$ with the entries $\A(x,y)$ is called \emph{tridiagonal} if $\A(x,y)=0$ for $|x-y|\ge2$. A \emph{Jacobi matrix} is a real symmetric tridiagonal matrix whose off-diagonal entries are strictly positive. We will deal with finite and semi-infinite Jacobi matrices; in the former case we assume that the indices $x,y$ range over $\{0,\dots,M\}$, where $M$ is a positive integer, and in the latter case we assume that $x,y$ range over $\Z_{\ge0}$. 

Thus, a finite Jacobi matrix is determined by the diagonal entries $\A(x,x)$ (where $x=0,\dots,M$) and the off-diagonal entries $\A(x,x+1)$ (where $n=0,\dots, M-1$). Likewise, a semi-infinite Jacobi matrix is determined by two infinite sequences $\{\A(x,x)\}$, $\{\A(x,x+1)\}$ indexed by $x\in\Z_{\ge0}$. 

Following \cite{A}, we call a semi-infinite Jacobi matrix a \emph{$\J$-matrix}. We also abbreviate $\ell^2:=\ell^2(\Z_{\ge0})$. 

Let $\ell^2_0$ denote the dense subspace of $\ell^2$ formed by the vectors with finitely many nonzero coordinates. Every $\J$-matrix $\A$ determines a symmetric operator on $\ell^2$ with domain $\ell^2_0$; we denote this operator again by $\A$. As is well known (see \cite{A}), the deficiency indices of $\A$ are either $(1,1)$ or $(0,0)$, and the question of which of these two cases holds is closely related to the classical moment problem. 

If the deficiency indices are $(0,0)$, then $\A$ is essentially self-adjoint, so that its closure $\overline\A$ is a self-adjoint operator and $\ell^2_0$ is a core for $\overline\A$. Here is a simple condition of essential self-adjointness which suffices for our purposes:

\begin{proposition}\label{prop4.D}
Let $\A$ be a $\J$-matrix such that
$$
\sum_{x=0}^\infty \frac1{\A(x,x+1)}=+\infty.
$$
Then its deficiency indices are $(0,0)$, and hence $\A$ is essentially self-adjoint.
\end{proposition}

\begin{proof}
See \cite{A}, Chapter I, Addenda and Problems, item 1.
\end{proof}

For a $\J$-matrix $\A$ with deficiency indices $(0,0)$, the domain of $\overline\A$ is easily described:

\begin{proposition}\label{prop4.E}
In the case of deficiency indices $(0,0)$, the domain of the self-adjoint operator $\overline\A$ consists of those vectors $v=(v_0,v_1,\dots)\in\ell^2$ for which the infinite vector $\A v$ with the coordinates
\begin{gather*}
(\A v)_0:=\A(0,1)v_1+\A (0,0)v_0,\\
(\A v)_n:=\A(n,n+1) v_{n+1}+\A(n,n) v_n+\A(n,n-1) v_{n-1}, \qquad n=1,2,\dots,
\end{gather*}
is still in $\ell^2$, and then $\overline\A v=\A v$.
\end{proposition}

\begin{proof}
The fact that $\A$ has deficiency indices $(0,0)$ just means that $\overline\A=\A^*$, and then one can apply a simple argument, see, e.g., \cite[Chapter IV, \S1.1]{A}.
\end{proof}

\subsection{The method}\label{sect4.3}

Now we are in a position to explain how our method works. We consider certain projection operators $K_1,K_2, \dots$, $K$ acting on $\ell^2$, and we want to show that $K_N\to K$ in the weak=strong operator topology. In the concrete cases under consideration we are able to exhibit essentially self-adjoint operators $\A_1,\A_2,\dots$, $\A$ with domain $\ell^2_0$, such that 
$$
K_N=[A_N]_+, \quad K=[A]_+, \quad \textrm{where} \quad N=1,2,\dots, \quad A_N:=\overline\A_N, \; A:=\overline\A
$$
(we recall that the symbol $[\,\cdot\,]_+$ denotes the spectral projection corresponding to the interval $(0,+\infty)$).  The operator $\A$ is given by a $\J$-matrix with deficiency indices $(0,0)$, and each $\A_N$  is given by a tridiagonal matrix which is either a $\J$-matrix with deficiency indices $(0,0)$ or the direct sum of a finite Jacobi matrix and a scalar matrix of infinite size. We verify that, as $N\to\infty$,
$$
\A_N(x,x)\to \A(x,x), \quad \A_N(x,x+1)\to\A(x,x+1) \quad \textrm{for any fixed index $x\in\Z_{\ge0}$}.
$$
This exactly means that $A_N\to A$ on $\ell^2_0$. Since $\ell^2_0$ is a common core, Proposition \ref{prop4.C} shows that  $A_N\to A$ in the strong resolvent sense. In our examples, the spectrum of $A$ is purely continuous. Therefore, we may apply Proposition \ref{prop4.B} and conclude that $K_N\to K$, as desired.

The point is that (in our examples)  the kernels of the projection operators are expressed through transcendental functions while the entries of the tridiagonal matrices are given by simple elementary formulas. For this reason, working with tridiagonal matrices turns out to be much easier than with kernels: asymptotic analysis reduces to elementary computations.

\subsection{Large-$N$ limits in variant 1: Charlier and Meixner ensembles}\label{sect4.4}

In Section \ref{sect6} we investigate five systems of discrete orthogonal polynomials: Charlier, Meixner, Krawtchouk, Hahn, and Racah. The necessary information about these polynomials is contained in Koekoek--Swarttouw \cite{KS}. In all cases, large-$N$ limit transitions are computed by the same algorithm; we proceed to its description. 

Let $P_0,P_1,\dots$ denote any of these systems of polynomials and $W$ denote the corresponding weight function. There are slight differences between two variants: 

\emph{Variant} 1: Charlier and Meixner.

\emph{Variant} 2: Krawtchouk, Hahn, and Racah.  

The reason is that in variant 1, the support of $W$ is the whole set $\Z_{\ge0}$ and there are infinitely many polynomials, while in variant 2, the support of $W$ is a finite set of the form $\{0,\dots,M\}\subset\Z_{\ge0}$ and  the system comprises  finitely many polynomials $P_0,\dots,P_M$ only. 

Let us examine variant 1 first.   

\smallskip 

\emph{Step} 1. We consider the weighted Hilbert space $\ell^2(\Z_{\ge0}, W)$ and observe that the space of polynomials $\C[x]$ is its dense subspace. Indeed, this claim is well known in the case of Charlier and Meixner polynomials. To verify it one can use the following general result: the space of polynomials  is dense if and only if  the moment problem related to the weight function $W$ is determinate (see, e.g., \cite[p. 131, Prop. 4.15]{Si}), which in turn is guaranteed if the exponential generating series for the moments of $W$ has a nonzero radius of convergence (\cite[p. 88, Prop. 1.5]{Si}). The fact that the latter property holds in the Charlier or Meixner case is easy to check.

\smallskip

\emph{Step} 2. Associated with $\{P_0,P_1,\dots\}$ is a  second order difference operator $\D$ with the following properties: 

$\bullet$ the $P_n$'s are its eigenfunctions, $\D P_n=-\mu_n P_n$,  and one has  
$$
0=\mu_0<\mu_1<\mu_2<\dots;
$$

$\bullet$ the action $\D$ on a test function $f(x)$ on $\Z_{\ge0}$ is given by a tridiagonal matrix:
$$
(\D f)(x)=\D(x,x+1) f(x+1)+\D(x,x) f(x)+\D(x,x-1)f(x-1),  \qquad x\in\Z_{\ge0};
$$

$\bullet$ the off-diagonal coefficients  $\D(x,x\pm1)$ are strictly positive (with the only exception of $\D(0,-1):=0$) and satisfy the relation
$$
W(x)\D(x,x+1)=W(x+1)\D(x+1,x)
$$
(it means that $\D$ is symmetric with respect to the inner product of $\ell^2(\Z_{\ge0},W)$;

$\bullet$ the diagonal coefficients are negative and given by
$$
\D(x,x)=-\D(x,x+1)-\D(x,x-1).
$$

The operator of multiplication by the function $\sqrt W$ establishes an isomorphism of Hilbert spaces $\ell^2(\Z_{\ge0}, W)\to\ell^2$. It transforms $\D$ into another second order difference operator $\wt\D:= W^{1/2}\circ\D\circ W^{-1/2}$. Its action is written as
$$
(\wt\D f)(x)=\wt\D(x,x+1) f(x+1)+\wt\D(x,x) f(x)+\wt\D(x,x-1)f(x-1),  \qquad x\in\Z_{\ge0},
$$
where
$$
\wt\D(x,x\pm1):=\sqrt{\frac{W(x)}{W(x\pm1)}}, \qquad \wt\D(x,x):=\D(x,x).
$$
Note that
$$
\wt\D(x,x\pm1)=\wt\D(x\pm1,x),
$$
so that $\wt\D$ is given by a Jacobi matrix. 

We check that the assumption of Proposition \ref{prop4.D} is satisfied, and hence the operator $\wt\D\big|_{\ell^2_0}$, the restriction of $\wt\D$ to $\ell^2_0$, is essentially self-adjoint. Let $D$ denote its closure; this is a self-adjoint operator on $\ell^2$.

On the other hand, since the polynomials $P_n$ are eigenfunctions of the difference operator $\D$ with eigenvalues $-\mu_n$, we have
$$
\wt\D (W^{1/2}P_n)=-\mu_n 
 W^{1/2}P_n, \qquad n=0,1,2,\dots\,.
$$
We claim that the same holds with $\wt\D$ replaced by $D$, that is,
$$
D (W^{1/2}P_n)=-\mu_n 
 W^{1/2}P_n, \qquad n=0,1,2,\dots\,.
$$
At first glance this looks evident, but actually is not, because $D$ is defined as the closure of the operator $\wt\D\big|_{\ell^2_0}$, while the functions $W^{1/2}P_n$ do not belong to $\ell^2_0$. But this difficulty is resolved with the help of Proposition \ref{prop4.E}: it tells us that the functions $W^{1/2}P_n$ lie in the domain of $D$, and on all these functions, the action of $D$ is implemented by $\wt\D$. 

We conclude that the self-adjoint operator $D$ is diagonalized in the orthogonal basis of $\ell^2$ formed by the functions $W^{1/2}P_n$, $n=0,1,\dots$\,. In particular, it has purely point, multiplicity free spectrum $0=-\mu_0>-\mu_2>\dots$\,. 

\smallskip

\emph{Step} 3.
Given $N=1,2,\dots$, we set
$$
A_N:=\frac1{c_N}(D+\mu_N), 
$$
where $c_N>0$ is an appropriate constant. 
The results of the previous steps show that the spectral projection $[D+\mu_N]_+=[D+\mu_N]_{(0,+\infty)}$ coincides with $K_N$, the $N$-dimensional projection operator introduced in Section \ref{sc:orth-poly-ens}. (Recall that the range of $K_N$ is the $N$-dimensional subspace of $\ell^2$ spanned by the first $N$ functions $W^{1/2}P_n$, $n=0,\dots,N-1$.) Division by a positive constant factor does not affect the spectral projection corresponding to the ray $(0,+\infty)$, so that we have
$$
[A_N]_+=K_N.
$$

On the other hand, we know that $A_N$ is determined by the Jacobi matrix $\A_N$ with the coefficients
\begin{equation}\label{eq4.A}
\begin{gathered}
\A_N(x,x+1):=\frac1{c_N}\left(\frac{W(x)}{W(x+1)}\right)^{\frac12} \D(x,x+1),\\
\A_N(x,x):=\frac{-\D(x,x+1)-\D(x,x-1)+\mu_N}{c_N}, \qquad x=0,1,2,\dots\,.
\end{gathered}
\end{equation}

Note that $W(x)$ and $\D(x,x\pm1)$ depend on the parameters entering the definition of the polynomials $P_n$.  Now our task is to find an appropriate limit regime: we let $N\to\infty$ and tune these parameters (which become depending on $N$) together with $c_N$ in such a way that there exist limits
$$
\A(x,x+1):=\lim_{N\to\infty}\A_N(x,x+1)>0, \quad \A(x,x):=\lim_{N\to\infty}\A_N(x,x), \qquad \forall x\in\Z_{\ge0}.
$$
Note that such a limit regime is not unique. In our examples, it depends on the deformation parameter $\r$. Moreover, as is seen from the results of Section \ref{sect6}, various limit regimes can differ in a more substantial way. 

\smallskip

\emph{Step} 4. Let $\A$ denote the limit Jacobi matrix. In our examples, it comes from the three-term relation for one of the three systems of orthogonal polynomials investigated in Section \ref{sect3}. Namely, let $\{\PP_n: n\in\Z_{\ge0}\}$ be the common notation for the Hermite, Laguerre, and Jacobi polynomials. We denote by  $\wt\PP_n$  the corresponding orthonormal polynomials and by $\W(dt)$ the weight measure. Next, let $T$ denote the operator of multiplication by the coordinate function $t$ acting on the Hilbert space $\H:=L^2(\R,\W(dt))$. The Jacobi matrix $\A$ turns out to be the matrix of the operator $\pm(T-\r)$ in the basis $\{\PP_n\}$ with an appropriate choice of $(\pm, \r)$. 

In Section \ref{sect5} below we write down the matrices $\A$ explicitly in each of the three cases. Their off-diagonal entries satisfy the assumption of Proposition \ref{prop4.D} and hence the deficiency indices of $\A$ are $(0,0)$. We denote by $A$ the self-adjoint operator on $\ell^2_0$ obtained by taking the closure of $\A\big|_{\ell^2_0}$. 

We claim that $A$ is precisely the image of $\pm(T-\r)$ under the isomorphism $L^2(\R,\W(dt))\to\ell^2$ taking the basis $\{\PP_n\}$ to the canonical basis of $\ell^2$. Indeed, to see this we argue as in the end of Step 2, with appeal to Proposition \ref{prop4.E}. 

Because the spectrum of operator $\pm(T-\r)$ is purely continuous, the same holds for $A$. Then the argument of Section \ref{sect4.3} shows that the spectral projections $[A_N]_+$ converge to the spectral projection $[A]_+$. This gives the final result: the convergence of the $N$-point ensemble under consideration to one of the ensembles from Section \ref{sect3}.

\subsection{Large-$N$ limits in variant 2: Krawtchouk, Hahn, and Racah}\label{sect4.5}

The algorithm of Section \ref{sect4.4} remains essentially the same, but due to finiteness of the support of the weight function, the situation is simplified and some of the arguments can be omitted. We only indicate necessary modifications. 

\smallskip

\emph{Step} 1. Instead of $\ell^2(\Z_{\ge0},W)$ we have to deal with the finite-dimensional space $\ell^2(\{0,\dots,M\}, W)$. It coincides with the linear span of $P_0,\dots,P_M$. No appeal to the moment problem is needed.

\smallskip

\emph{Step} 2. The infinite system of orthogonal polynomial is replaced by a finite one, $\{P_0,\dots,P_M\}$, and the Hilbert space $\ell^2$ is replaced by the $(M+1)$-dimensional Hilbert space $\ell^2(\{0,\dots,M\})$. Because of this the situation is simplified (no need to take the closure of a densely defined operator). 

\smallskip

\emph{Step} 3. 
We define the Jacobi matrix $\A_N$ as before, but now it has finite format $(M+1)\times(M+1)$. The number $M$ is an additional parameter which will vary together with $N$; we assume $N\le M$, so $M$ will grow with $N$. Since the spaces $\ell^2(\{0,\dots,M\})$ will vary, it is convenient to extend the pre-limit operators $A_N$ to the space $\ell^2$.  To do this we consider the natural direct sum decomposition
$$
\ell^2=\ell^2(\{0,\dots.M\})\oplus\ell^2(\{M+1,M+2,\dots\})
$$ 
and set $A_N=-1$ on the second component; this does not affect the spectral projection $[A_N]_+$.

\smallskip

\emph{Step} 4. Here nothing changes.

\section{Jacobi matrices associated with ensembles $\DH^\pm(\r)$, $\DL^\pm(\r;\be)$, and $\DJ^\pm(\r;a,b)$}\label{sect5}

\subsection{Jacobi matrices}
The purpose of the present section is to exhibit $\J$-matrices which will appear in various concrete instances of limit transitions.

We begin with a general definition, where we adopt the notation and assumptions of Section \ref{sect3.1}. Thus, $\{\PP_x: x\in\Z_{\ge0}\}$ is a system of orthogonal polynomials with a weight measure $\W(dt)$ satisfying conditions (i)--(iii), $\wt\PP_x$ are the corresponding orthonormal polynomials, and $\r$ is a fixed point inside the support of $\W$. Recall that from these data we can construct determinantal point processes $\P^+_\r$ and $\P^-_\r$.

Let $T$ denote the operator on $L^2(\R,\W(dt))$ consisting in multiplication by the coordinate function $t$. This is a self-adjoint operator. Its action in the basis $\{\PP_n\}$ is given by a semi-infinite tridiagonal matrix:
$$
T\PP_x= T(x+1,x)\PP_{x+1}+ T(x,x)\PP_x+T(x-1,x)\PP_{x-1};
$$
this is a reformulation of the classical three-term relation which holds for any system of orthogonal polynomials.

In the orthonormal basis $\{\wt\PP_x\}$, the action of $T$ can be written as 
$$
T\wt\PP_x= \wt T(x+1,x)\wt\PP_{x+1}+ \wt T(x,x)\wt\PP_x+\wt T(x-1,x)\wt\PP_{x-1},
$$
where
$$
\wt T(x\pm1,x)=\frac{\Vert\PP_{x\pm1}\Vert}{\Vert\PP_x\Vert} T(x\pm1,x), \quad \wt T(x,x)=T(x,x).
$$
Note that the matrix $\wt T$ is symmetric: $\wt T(x\pm1,x)=\wt T(x,x\pm1)$. 

\begin{definition}\label{def5.A}
We define the symmetric tridiagonal matrix $\A=\A^\pm_\r$ via
$$
\A(x,x+1)=\A(x+1,x):=\wt T(x+1,x), \quad \A(x,x):=\pm(\wt T(x,x)-\r).
$$
In other words, $\A^+_\r$ is the matrix of the operator $T-\r$ in the basis $\{\wt\PP_x\}$, and $\A^-_\r$ is the matrix of the operator $-(T-\r)$ in the basis $\{(-1)^x\wt\PP_x\}$. 
\end{definition}

\begin{proposition}\label{prop5.A}
Let $\A=\A^\pm_\r$ be the matrix just defined.

{\rm(i)} $\A(x,x+1)>0$ for all $x\in\Z_{\ge0}$, so $\A$ is a $\J$-matrix.

{\rm(ii)}  $\A$ has deficiency indices $(0,0)$, whence it determines an essentially self-adjoint operator with domain $\ell^2_0$. 

{\rm(iii)} The self-adjoint operator $A:=\overline\A$ on $\ell^2$  has simple, purely continuous spectrum filling the support of $\W$.

{\rm(iv)} The determinantal measure corresponding to the spectral projection $[A]_+$ coincides with $\P^\pm_\r$. 
\end{proposition}

\begin{proof}
Note that of these four assertions, only (ii) is nonevident. 

(i) Recall that in our standardization, the leading coefficients of polynomials $\PP_n$ are strictly positive. This implies that $T(x+1,x)>0$ for all $x\in\Z_{\ge0}$, which in turn means that $\wt T(x+1,x)>0$ and so $\A(x,x+1)=\A(x+1,x)>0$. Thus, $\A$ is a $\J$-matrix.

(ii) By the definition of $\A$, (ii) can be rephrased as follows:  the restriction of the operator $T$ to the subspace $\C[t]\subset L^2(\R,\W(dt))$ is essentially self-adjoint. It is known (see, e.g., \cite[p. 86, Theorem 2]{Si}) that the latter fact holds if and only if the moment problem for $\W$ is determinate. But this property has been postulated. 

Alternatively, in the case of Hermite, Laguerre or Jacobi polynomials, when we dispose of explicit expressions for the entries $\A(x,x+1)$ (see below), we can deduce (ii) from Proposition \ref{prop4.D}.

(iii) Evident, because $A$ is equivalent to the operator $\pm(T-\r)$.

(iv) The determinantal measure in question has $\A(x,y)$ as a correlation kernel. On the other hand, according to Definition \ref{def3.B}, the correlation kernel $K^\pm_\r(x,y)$ of the measure $\P^\pm_\r$ is the matrix of $\pm(T-\r)$ in the basis $\{\wt\PP_x\}$. Hence, 
$$
\A^+_\r(x,y)=K^+_\r(x,y) \quad \text{and}\quad \A^-_\r(x,y)=(-1)^{x+y}K^-_\r(x,y).
$$
Thus, in the case $\A=\A^+_\r$, our assertion is trivial. In the case $\A=\A^-_\r$, we use the fact that the factor $(-1)^{x+y}$ does not affect the determinantal measure (see the end of Section \ref{sect2.1}).
\end{proof}

Below we write down the matrix $\A$ for three systems of orthogonal polynomials: Hermite, Laguerre, and Jacobi. This is a simple exercise: we use the formulas above and the explicit expression for the three-term relation, which we take from \cite{KS}. 

\subsection{The Jacobi matrix associated with the discrete Hermite ensemble  $\DH^\pm(r)$}
The three-term relation for the Hermite polynomials has the form (see \cite[(1.13.3)]{KS})
$$
tH_x=\frac12 H_{x+1}+x H_{x-1}.
$$
Together with the expression for the norm, see Section \ref{sect3.2}, this implies
\begin{equation}\label{Hermite}
\A(x,x+1)=\sqrt{\frac{x+1}2}, \quad \A(x,x)=\mp\r.
\end{equation}

\subsection{The Jacobi matrix associated with the discrete Laguerre ensemble  $\DL^\pm(r;\be)$}
The three-term relation for the Laguerre polynomials in our standardization has the form (see \cite[(1.11.3)]{KS})
$$
tL^{(\be)}_x=(x+1) L^{(\be)}_{x+1}+(2x+\be) L^{(\be)}_x+(x+\be-1)L^{(\be)}_{x-1}.
$$
Together with the expression for the norm (see Section \ref{sect3.3}), we obtain
\begin{equation}\label{Laguerre}
\A(x,x+1)=\sqrt{(x+1)(x+\be)}, \quad \A(x,x)=\pm(2x+\be-\r).
\end{equation}

\subsection{The Jacobi matrix associated with the discrete Jacobi ensemble  $\DJ^\pm(r;a,b)$}
The three-term relation for the Jacobi polynomials has the form (see \cite[(1.8.3)]{KS})
\begin{gather*}
tP^{(a,b)}_x=\frac{2(x+1)(x+a+b+1)}{(2x+a+b+1)(2x+a+b+2)} P^{(a,b)}_{x+1}\\
+\frac{b^2-a^2}{(2x+a+b)(2x+a+b+2)}P^{(a,b)}_x
+\frac{2(x+a)(x+b)}{(2x+a+b)(2x+a+b+1)}P^{(a,b)}_{x-1}.
\end{gather*}
Together with the expression for the norm, see Section \ref{sect3.4}, this implies
\begin{gather}
\A(x,x+1)=\frac2{2x+a+b+2}\left\{\frac{(x+1)(x+a+1)(x+b+1)(x+a+b+1)}{(2x+a+b+1)(2x+a+b+3)}\right\}^{1/2}
\label{Jacobi1}\\
\A(x,x)=\pm\left(\frac{b^2-a^2}{(2x+a+b)(2x+a+b+2)}-\r\right). \label{Jacobi2}
\end{gather}

\section{Large-$N$ limit transitions: concrete computations}\label{sect6}

In this section we establish several limit transitions between determinantal point processes that can be seen in Figure \ref{fig}; see Sections \ref{sect4.4} and \ref{sect4.5} for a general description of how the computations proceed. 

\begin{figure}[htbp]
\centering
\includegraphics[scale=1]{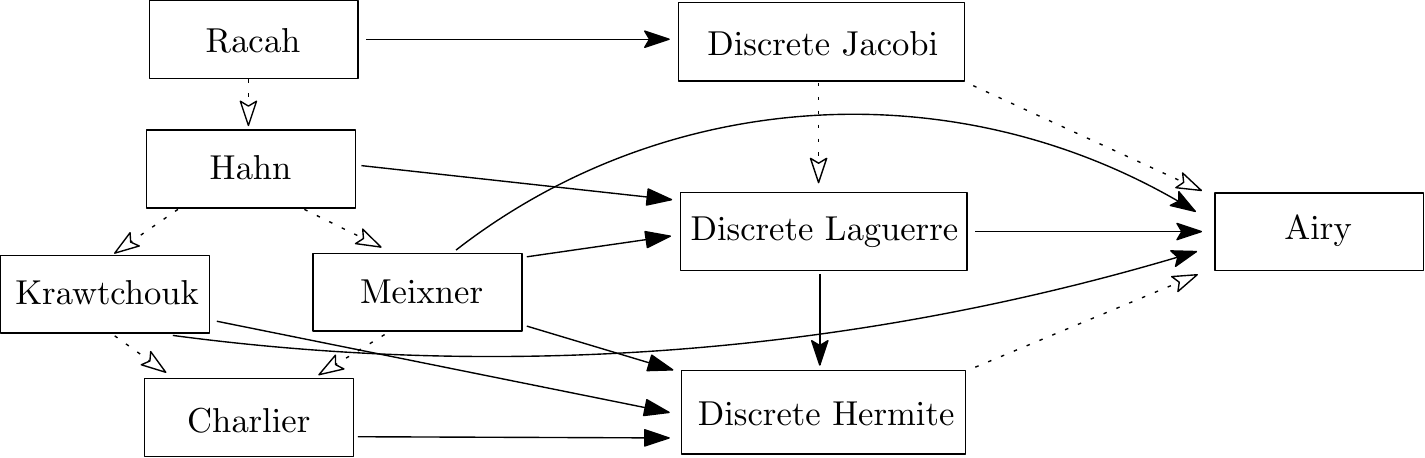}
\caption{A scheme of limit transition between determinantal point processes. A single name associated with classical orthogonal polynomials stands for the corresponding orthogonal polynomial ensembles. Solid arrows depict limits that are explicitly discussed in the text. Dotted arrows denote transitions that are not hard to obtain along the same lines, but that are not explicitly discussed. \label{fig}}
\end{figure}

\subsection{$N$-particle Charlier ensembles $\to$ discrete Hermite ensemble}

The \emph{Charlier polynomials} are orthogonal polynomials on $\Z_{\ge0}$ with the weight function 
$$
W(x)=\frac{\th^x}{x!}, \qquad x\in\Z_{\ge0},
$$ 
depending on a parameter $\th>0$, see \cite[\S1.12]{KS}.  We denote by $\Charlier(N,\th)$ the corresponding $N$-particle ensemble. We need the discrete Hermite ensemble $\DH^+(\r)$, which is defined in Section \ref{sect3.2}. 

\begin{theorem}
Fix $\r\in\R$, let $N\to\infty$, and let the parameter $\th$ vary together with $N$ in such a way that 
$$
\frac{\th-N}{\sqrt N}\to \sqrt2\r. 
$$
In this limit regime, $\Charlier(N,\th)\to\DH^+(\r)$. 
\end{theorem}

\begin{proof}
We apply the algorithm described in Section \ref{sect4.4}, where we use the above expression for the weight function and the following formulas taken from  \cite[(1.12.5)]{KS}):
$$
\D(x,x+1)=\th, \quad \D(x,x-1)=x, \quad \mu_n=n,
$$
where $x$ and $n$ range over $\Z_{\ge0}$.

From these data we compute the entries of the Jacobi matrix $\A_N$ according to \eqref{eq4.A}:
$$
\A_N(x,x+1)=\frac{\sqrt{\th(x+1)}}{c_N}, \qquad 
\A_N(x,x)=\frac{-\th-x+N}{c_N}, \qquad x\in\Z_{\ge0}.
$$
Let us take $c_N:=\sqrt{2N}$. Then, as $N\to\infty$, 
$$
\A_N(x,x+1)\to \sqrt{\frac{x+1}2}, \qquad \A_N(x,x)\to -\r.
$$
The limit values coincide with the entries $\A(x,x+1)$, $\A(x,x)$ of the Jacobi matrix corresponding to the discrete Hermite ensemble $\DH^+(\r)$, see \eqref{Hermite}. This completes the proof.

\end{proof}

\subsection{$N$-particle Meixner ensembles $\to$ discrete Laguerre ensemble}\label{sc:6.2}

The \emph{Meixner polynomials} are orthogonal polynomials on $\Z_{\ge0}$ with the weight function
$$
W(x)=\frac{(\be)_x}{x!}\xi^x=\frac{\Ga(\be+x)}{\Ga(\be)x!}\xi^x, \qquad x\in\Z_{\ge0},
$$
which depends on two parameters, $\be>0$ and $\xi\in(0,1)$, see \cite[\S1.9]{KS} (in \cite{KS}, the parameter $\xi$ is denoted by $c$). We denote by $\Meixner(N,\be,\xi)$ the corresponding $N$-particle ensemble. We also need the discrete Laguerre ensemble $\DL^-(\r)$, which is defined in Section \ref{sect3.3}.

\begin{theorem}\label{thm4.A}
Fix $s>0$ and $\be>0$, let $N\to\infty$, and let the parameter $\xi$ vary together with $N$ in such a way that $(1-\xi)N\to \r$, that is,
$$
\xi=1-\frac \r N+o(N).
$$
In this limit regime,
$\Meixner(N,\be,\xi)\to\DL^-(\r;\be)$.  
\end{theorem}

\begin{proof}
We apply the algorithm of Section \ref{sect4.4},  for which we use the above expression for the weight function and the following formulas taken from   \cite[(1.9.5)]{KS}:
$$
\D(x,x+1)=\xi(x+\be), \quad \D(x,x-1)=x, \quad \mu_n=(1-\xi)n,
$$
where $x$ and $n$ range over $\Z_{\ge0}$.

From these data we compute the entries of the Jacobi matrix $\A_N$ according to \eqref{eq4.A}:
\begin{equation}\label{eq6.A}
\A_N(x,x+1)=\frac{\sqrt{\xi(x+1)(x+\be)}}{c_N},  \qquad 
\A_N(x,x)=\frac{-\xi(x+\be)-x+(1-\xi)N}{c_N}.
\end{equation}
Let us take $c_N:=1$. Then, as $N\to\infty$, 
$$
\A_N(x,x+1)\to \sqrt{(x+1)(x+\be)}, \qquad \A_N(x,x)\to -(2x+\be)+ \r,
$$
and  the limit values coincide with the entries $\A(x,x+1)$, $\A(x,x)$ of the Jacobi matrix corresponding to the discrete Laguerre ensemble $\DL^-(\r;\be)$, cf. \eqref{Laguerre}.

\end{proof}

\subsection{$N$-particle Meixner ensembles $\to$ discrete Hermite ensemble}\label{sc:6.3}

Here we exhibit another limit regime for the Meixner ensembles, which leads to the discrete Hermite ensemble $\DH^+(\r)$ defined in Section \ref{sect3.1}.

\begin{theorem}\label{th:6.3}
Fix $\r\in\R$, let $N\to+\infty$, and let the parameters $\be$ and $\xi$ vary together with $N$ in such a way that 
\begin{equation}\label{eq:th6.3}
\xi\be\to+\infty, \qquad -\sqrt{\xi\be}+\frac{(1-\xi)N}{\sqrt{\xi\be}}\to -\sqrt2\r.
\end{equation}

In this limit regime,
$\Meixner(N,\be,\xi)\to\DH^+(\r)$.  
\end{theorem}

\begin{proof}
We assume that $\be$ and $\xi$ depend on $N$ as indicated above, and we set in \eqref{eq6.A}
$$
c_N:=\sqrt{2\xi\be}.
$$
It follows that, as $N\to +\infty$,
$$
\A_N(x,x+1)\to\sqrt{\frac{x+1}2}, \qquad \A_N(x,x)\to -\r,
$$
and the limit values coincide with the entries of the Jacobi matrix corresponding to $\DH^+(\r)$, cf. \eqref{Hermite}.  
\end{proof}

\subsection{$N$-particle Krawtchouk ensembles $\to$ discrete Hermite ensemble}\label{sc:6.4}

The \emph{Krawtchouk polynomials} are orthogonal polynomials on $\{0,\dots,M\}$ with the weight function
$$
W(x)=\binom Mx p^x(1-p)^{M-x}, \qquad x=0,1,\dots,M,
$$
where $p\in(0,1)$ is a parameter, see \cite[\S1.10]{KS}. We denote by $\Krawtchouk(N,p,M)$ the corresponding $N$-particle ensemble, where we assume that $N\le M$ (see Section \ref{sc:orth-poly-ens}).

\begin{theorem}\label{th:6.4}
Fix $\r\in\R$, let $N\to+\infty$, and let the parameters $p$ and $M$ vary together with $N$ in such a way that 
\begin{equation}\label{eq:th6.4}
pM\to+\infty, \qquad -\sqrt{pM}+\frac{N}{\sqrt{pM}}\to -\sqrt2\r.
\end{equation}

In this limit regime,
$\Krawtchouk(N,p,M)\to\DH^+(\r)$. 
\end{theorem}

Note that the restriction $N\le M$ is not violated: indeed,  $pM\asymp N^2$ and $p\in(0,1)$, so that $M$ grows faster than $N$.

\begin{proof}
We apply the algorithm described in Section \ref{sect4.3} with the modifications indicated in Section \ref{sect4.4}, and  we use the above expression for the weight function and the following expressions taken from   \cite[(1.10.5)]{KS}):
$$
\D(x,x+1)=p(M-x), \quad \D(x,x-1)=(1-p)x, \quad \mu_n=n,
$$
where both $x$ and $n$ range over $\{0,\dots,M\}$.

From these data we compute the entries of the Jacobi matrix $\A_N$ according to \eqref{eq4.A}:
$$
\A_N(x,x+1)=\frac{\sqrt{p(x+1)(M-x)}}{c_N}, \qquad
\A_N(x,x)=\frac{-p(M-x)-(1-p)x+N}{c_N}
$$
Let us take $c_N:=\sqrt{2pM}$. Then, as $N\to\infty$, 
$$
\A_N(x,x+1)\to \sqrt{\frac{x+1}2}, \qquad \A_N(x,x)\to -\r,
$$
and the limit values coincide with the entries $\A(x,x+1)$, $\A(x,x)$ of the Jacobi matrix corresponding to $\DH^+(\r)$, see \eqref{Hermite}.
\end{proof}

Note a similarity with the above computation for the Meixner polynomials. This is not surprising because of a well-known relation between the Meixner and Krawtchouk polynomials, see the last formula in \cite[\S1.10]{KS}.

\subsection{$N$-particle Hahn ensembles $\to$  discrete Laguerre ensemble}

The \emph{Hahn polynomials} are orthogonal polynomials on $\{0,\dots,M\}$ with the weight function
$$
W(x)=\binom{a+x}{x}\binom{b+M-x}{M-x}=\frac{\Ga(a+x+1)\Ga(b+M-x+1)}{x!\Ga(a+1)(M-x)!\Ga(b+1)}, \quad x=0,1,\dots,M,
$$
depending on parameters $a>-1$ and $b>-1$, see \cite[\S1.5]{KS} (warning: in \cite{KS},  our triple of parameters $(a,b,M)$ is denoted as $(\al,\be,N)$). We denote by $\Hahn(N, a,b,M)$ the corresponding $N$-particle ensemble, where we assume that $N\le M$. 

\begin{theorem}
Fix $\r>0$, $a>-1$, and $b>-1$. Let $N\to+\infty$ and let the parameter $M$ vary together with $N$ in such a way that  $M\sim \r^{-1}N^2$.

In this limit regime, $\Hahn(N,a,b,M)\to\DL^-(\r; a+1)$. 

\end{theorem}

\begin{proof}
We apply the algorithm described in Section \ref{sect4.3} with the modifications indicated in Section \ref{sect4.4}, and  we use the above expression for the weight function and the following expressions taken from   \cite[(1.5.5)]{KS}:
\begin{gather*}
\D(x,x+1)=(M-x)(x+a+1), \quad \D(x,x-1)=x(M+b+1-x), \\
\mu_n=n(n+a+b+1),
\end{gather*}
where both $x$ and $n$ range over $\{0,\dots,M\}$. Note that $0=\mu_0<\mu_1<\dots<\mu_M$, as required.

From these data we compute the entries of the Jacobi matrix $\A_N$ according to \eqref{eq4.A}:
\begin{gather*}
\A_N(x,x+1)=\frac{\sqrt{(M-x)(x+1)(x+a+1)(M+b-x)}}{c_N},\\
\A_N(x,x)=\frac{-(M-x)(x+a+1)-x(M+b+1-x)+N(N+a+b+1)}{c_N}.
\end{gather*}
Let us take $c_N:=M$. Then, as $N\to\infty$, 
$$
\A_N(x,x+1)\to \sqrt{(x+1)(x+a+1)}, \qquad \A_N(x,x)\to -(2x+a+1)+\r,
$$
and the limit values coincide with the entries $\A(x,x+1)$, $\A(x,x)$ of the Jacobi matrix corresponding to $\DL^-(\r,a+1)$, cf. \eqref{Laguerre}. 
\end{proof}

\subsection{$N$-particle Racah ensembles $\to$ discrete Jacobi ensemble}\label{Racah-DJacobi}

The \emph{Racah polynomials} are at the top of the Askey scheme, see \cite[\S1.2]{KS}. They depend on a quadruple of parameters $(\al,\be,\ga,\de)$, whose range splits into several pieces. For our purposes it suffices to choose one of them; namely, in what follows we assume that  
$$
\al+1=-M, \quad \be>M+\ga, \quad \ga>-1,\quad \de>-1,
$$ 
where $M$ is a positive integer. We regard the Racah polynomials as functions of the variable $x$ ranging over the set $\{0,\dots,M\}$, but it should be noted that they are actually orthogonal polynomials on the \emph{quadratic grid} 
$$
\X_M:=\{x(x+\ga+\de+1): x=0,1,\dots,M\}.
$$
Thus, the $n$th Racah polynomial is a polynomial of degree $n$ with respect to the variable $\wt x:=x(x+\de+\ga+1)$, and not the variable $x$, as before. However, this does not affect our algorithm,  because the map $x\mapsto \wt x$ defines a bijection between $\{0,\dots,M\}$ and $\X_M$. As in the case of Krawtchouk or Hahn polynomials, there are $M+1$ Racah polynomials, they are linearly independent as functions on $\{0,\dots,M\}$, and their linear span is precisely the space of all functions on this set.  

The weight function is given by
$$
W(x)=\frac{(\al+1)_x(\be+\de+1)_x(\ga+1)_x(\ga+\de+1)_x(\frac{\ga+\de+3}2)_x}{(-\al+\ga+\de+1)_x(-\be+\ga+1)_x(\frac{\ga+\de+1}2)_x(\de+1)_x x!},
$$
see \cite[(1.2.2)]{KS}; here 
$$
(a)_x:=a(a+1)\dots(a+x-1)=\frac{\Ga(a+x)}{\Ga(a)}
$$
is the standard notation for the Pochhammer symbol.

The corresponding $N$-particle orthogonal polynomial ensemble is denoted by $\Racah(N,\alde)$, where $N\le M$. Recall that $\al$ and $M$ are related by $\al+1=-M$; below we alternately use one of these two parameters. 

\begin{theorem}
Fix the parameters $a>-1$, $b>-1$, and $\r\in(-1,1)$, and consider the following limit regime{\rm:}
\begin{gather*}
N\to+\infty,\quad M\to+\infty, \quad \frac NM\to\sqrt{\frac{1-\r}2},\\
 \ga=a, \quad \de=b,\quad 
\al=-M-1, \quad \be=M+a+\const, \quad \const>0.
\end{gather*}

In this regime, $\Racah(N,\alde)\to\DJ^+(\r;a,b).$
\end{theorem}

\begin{proof}
We use the following formulas (see \cite[(1.2.5)]{KS}):
\begin{gather*}
\D(x,x+1)=\frac{(M-x)(x+\be+\de+1)(x+\ga+1)(x+\ga+\de+1)}{(2x+\ga+\de+1)(2x+\ga+\de+2)},\\
\D(x,x-1)=\frac{x(x+M+\ga+\de+1)(\be-\ga-x)(x+\de)}{(2x+\ga+\de)(2x+\ga+\de+1)},\\
\mu_n=n(n+\al+\be+1)=n(n+\be-M).
\end{gather*}

We claim that  $\D(x,x+1)>0$ for $x=0,1,\dots,M-1$. Indeed, 
due to the assumptions on the parameters, in the expression for $\D(x,x+1)$, all the factors are strictly positive for $x=1,\dots,M-1$. Next, if $x=0$, then the factor $x+\ga+\de+1$ in the numerator is cancelled with the factor $2x+\ga+\de+1$ in the denominator, and the remaining factors are again stirctly positive. 

Likewise, $\D(x,x-1)>0$ for $x=1,2,\dots,M$, because all factors are strictly positive. 
We also have $\D(M,M+1)=\D(0,-1)=0$ due to the factors $M-x$ and $x$, respectively.
Finally, we have $0=\mu_0<\mu_1<\dots<\mu_M$. 

Thus, all necessary conditions that are required in our algorithm are satisfied. 

\begin{lemma}
Let  $\A_N(x,x+1)$ and $\A_N(x,x)$ be defined according to \eqref{eq4.A}, and set $c_N=M^2/2$. 
Then, in the limit regime specified in the formulation of the theorem, 
$$
\A_N(x,x+1)\to\frac2{2x+a+b+2}\left\{\frac{(x+1)(x+a+1)(x+b+1)(x+a+b+1)}{(2x+a+b+1)(2x+a+b+3)}\right\}^{1/2}
$$
and
$$
\A_N(x,x)\to \frac{b^2-a^2}{(2x+a+b)(2x+a+b+2)}-\r.
$$
\end{lemma}

\begin{proof}
(i) From the expression for the weight function we obtain (recall that $\al=-M-1$)
\begin{multline*}
\sqrt{\frac{W(x)}{W(x+1)}}=\left\{\frac{(x-\al+\ga+\de+1)(x-\be+\ga+1)}{(x+\al+1)(n+\be+\de+1)}\right\}^{1/2}\\
\times\left\{\frac{(x+1)(x+\de+1)(2x+\ga+\de+1)}{(x+\ga+1)(x+\ga+\de+1)(2x+\ga+\de+3)}\right\}^{1/2}.
\end{multline*}
As $M\to\infty$, the first quantity in the curly brackets tends to 1 and hence
$$
\sqrt{\frac{W(x)}{W(x+1)}}\sim\left\{\frac{(x+1)(x+\de+1)(2x+\ga+\de+1)}{(x+\ga+1)(x+\ga+\de+1)(2x+\ga+\de+3)}\right\}^{1/2}.
$$

Next, from the expression for $\D(x,x+1)$ we obtain
$$
\D(x,x+1)\sim M^2\,\frac{(x+\ga+1)(x+\ga+\de+1)}{(2x+\ga+\de+1)(2x+\ga+\de+2)}.
$$

Since $c_N=M^2/2$, this implies the first formula.

(ii) Using the relation $\al+1=-M$ we obtain
\begin{multline*}
\D(x,x)=-\frac{(M-x)(x+\be+\ga+1)(x+\ga+1)(x+\ga+\de+1)}{(2x+\ga+\de+1)(2x+\ga+\de+2)}\\
-\frac{x(x+M+\ga+\de+1)(\be-\ga-x)(x+\de)}{(2x+\ga+\de)(2x+\ga+\de+1)}\\
\sim -M^2\left\{\frac{(x+\ga+1)(x+\ga+\de+1)}{(2x+\ga+\de+1)(2x+\ga+\de+2)}
+\frac{x(x+\de)}{(2x+\ga+\de)(2x+\ga+\de+1)}\right\}.
\end{multline*}
Therefore,
$$
\frac{\D(x,x)}{c_N}\to-\frac{2(x+\ga+1)(x+\ga+\de+1)}{(2x+\ga+\de+1)(2x+\ga+\de+2)}
-\frac{2x(x+\de)}{(2x+\ga+\de)(2x+\ga+\de+1)}.
$$
One readily checks that the right-hand side is equal to
$$
-1+\frac{\de^2-\ga^2}{(2x+\ga+\de)(2x+\ga+\de+2)}=-1+\frac{b^2-a^2}{(2x+a+b)(2x+a+b+2)}.
$$

Next,
$$
\frac{\mu_N}{c_N}=\frac{2N(N+\be-M)}{M^2}\sim\frac{2N^2}{M^2}\sim 1-\r.
$$

Adding up the two expressions gives the second formula.
\end{proof} 

The limit values computed in the lemma coincide with the entries of the Jacobi matrix associated with $\DJ^+(\r;a,b)$, see \eqref{Jacobi1} and \eqref{Jacobi2}. This completes the proof of the theorem.
\end{proof}

\subsection{Discrete Laguerre $\to$ discrete Hermite}\label{sc:6.7}
Here we find conditions under which the discrete Laguerre ensemble $\DL^\pm(\wt{\r},\be)$ converges to the discrete Hermite ensemble $\DH^\pm(\r)$. Recall that both parameters of the discrete Laguerre (here denoted as $\wt{\r}$, $\be$) should be positive reals. The setting of the problem differs from that of the preceding sections, because this is not a large-N limit transition. But the method is the same, with obvious simplifications. 

\begin{theorem}\label{thm6.LH}
We fix $\r\in\R$ and consider the limit regime for $\DL^\pm(\wt{\r},\be)$ in which $\be\to+\infty$ and $\wt{\r}=\wt{\r}(\be)$ varies in such a way that 
$$
 \frac{\be-\wt{\r}}{\sqrt{2\be}}\to-\r.
$$
Then $\DL^\pm(\wt{\r},\be)\to\DH^\pm(\r)$.
\end{theorem}

Recall that $\DH^-(\r)=\DH^+(-r)$, so that we can equally well achieve the convergence $\DL^\pm(\wt{\r},\be)\to\DH^\mp(\r)$.

\begin{proof} 
Let $\wt{\A}$ and $\A$ denote the Jacobi matrices associated with $\DL^\pm(\wt{\r},\be)$ and $\DH^\pm(\r)$, respectively. It suffices to show that 
$$
\frac{\wt{\A}(x,x+1)}{c(\wt\r, \be)}\to\A(x,x+1), \quad \frac{\wt{\A}(x,x)}{c(\wt\r, \be)}\to\A(x,x), \qquad \forall x\in\Z_{\ge0},
$$
with an appropriate choice of quantities $c(\wt\r, \be)$ depending on $\wt{\r}$ and $\be$. Recall (see \eqref{Laguerre} and \eqref{Hermite}) that 
\begin{gather*}
\wt{\A}(x,x+1)=\sqrt{(x+1)(x+\be)}, \quad \wt{\A}(x,x)=\pm(2x+\be-\wt{\r}),\\
\A(x,x+1)=\sqrt{\frac{x+1}2}, \quad \A(x,x)=\mp\r.
\end{gather*}
It follows that the desired limits hold provided we set $c(\wt\r, \be)=\sqrt{2\be}$. 
\end{proof}

\section{Convergence to the Airy ensemble}

\subsection{The Airy ensemble}
Consider the Airy differential operator 
$$
D^\Airy:=\frac{d^2}{dv^2}-v, 
$$
where the variable $v$ ranges over $\R$. This differential operator is essentially self-adjoint on $C^\infty_0(\R)$, so that its closure is a self-adjoint operator. The latter operator will be denoted by $A$. It has simple, purely continuous spectrum filling the whole axis $\R$. The spectral projection $K:=[A]_+$ gives rise to a determinantal point process, called the \emph{Airy ensemble}. We will denote it as $\Airy$. 

Below we discuss limit transitions from discrete ensembles on $\Z_{\ge0}$ to the Airy ensemble. This inevitably assumes a scaling of the grid $\Z_{\ge0}$. We will consider the scaling of the form
$$
x\mapsto v:=\frac{\si-x}\tau, 
$$
where $x\in\Z_{\ge0}$ is the discrete variable, $v\in\R$ is the continuous variable, and $\si>0$ and $\tau>0$ are large scaling parameters such that $\tau\ll\si$. According to this, given a point process $\P$ on $\Z_{\ge0}$, we denote by $\dfrac{\si-\P}\tau$ the transformed pre-limit process living on the grid
$$
\left\{\frac{\si-x}\tau: x\in\Z_{\ge0}\right\}\subset (-\infty, \si/\tau]\subset\R
$$
with mesh $\tau^{-1}$. As $\si$ and $\tau$ get large, the mesh of the grid tends to zero and its right end $\si/\tau$ shifts to $+\infty$, so in the limit the grid fills the whole real line.

\subsection{Discrete Hermite ensemble $\to$ Airy ensemble}

The result can be stated for the two variants, $\DH^+(\r)$ and $\DH^-(\r)$, simultaneously. We assume that $\r$ tends to $+\infty$ or $-\infty$, respectively, and we set
\begin{equation}\label{eq7.F}
\si=\si(r):=2^{-1} |\r|^2, \quad \tau=\tau(r):=\si^{1/3}=2^{-1/3}|\r|^{2/3}.
\end{equation}
We present a heuristic argument showing that in this limit regime,
$$
\frac{\si-\DH^\pm(\r)}{\tau} \to \Airy.
$$

Because of the symmetry $\DH^-(-\r)=\DH^+(\r)$, it suffices to examine the case of $\DH^+(\r)$. 
Let us denote by $D^{\DH^+(\r)}$ the difference operator on $\Z_{\ge0}$ that is defined by the Jacobi matrix associated with $\DH^+(\r)$. The action of $D^{\DH^+(\r)}$  on a test function $f$ is given by
\begin{equation}\label{eq7.A}
(D^{\DH^+(\r)})f(x)=\sqrt {\frac{x+1}2}f(x+1)-\r f(x)+\sqrt{\frac x2}f(x-1).
\end{equation}

Recall we have the freedom of multiplying the pre-limit operator by a positive constant depending on our large parameter $\r$. 
 
\begin{proposition}\label{prop7.A}
Fix an arbitrary smooth function $g(v)$ on $\R$ and assign to it a function $f(x)$ on $\Z_{\ge0}$ by setting 
$$
f(x)=g(v), \quad v=\frac{\si-x}\tau.
$$
Let $\si$ and $\tau$ depend on $\r$ as in \eqref{eq7.F},  and set 
$$
c=c(r):=\si^{1/6}=2^{-1/6}\r^{1/3}.
$$
As $\r\to+\infty$, we have
$$
\sqrt2\, c(r) D^{\DH^+(\r)}f(x) \to D^\Airy g(v).
$$
\end{proposition}

\begin{proof}
We expand
$$
f(x\pm1)=g(v\mp\tau^{-1})=g(v)\mp\tau^{-1}g'(v)+\frac12 \tau^{-2}g''(v)+\dots
$$
and substitute this in \eqref{eq7.A}. Then we obtain
\begin{multline*}
\sqrt 2\,D^{\DH^+(\r)}f(x)=(\sqrt{x+1}+\sqrt x-\sqrt2\,\r)g(v)\\
-(\sqrt{x+1}-\sqrt x)\tau^{-1} g'(v)
+\frac12(\sqrt{x+1}+\sqrt x)\tau^{-2}
g''(v)+\dots
\end{multline*}
Next, we substitute $x=\si-\tau v$ and use the expansion
$$
\sqrt{\si+a}=\si^{1/2}\left(1+\frac12 a\si^{-1}+o(a\si^{-1})\right)\quad \textrm{for $|a|\ll\si$}.
$$
This gives us:

$\bullet$ The coefficient of $g(v)$ is equal to 
\begin{multline*}
\sqrt{\si-\tau v+1}+\sqrt{\si+\tau v}-\sqrt2\,\r
=2\si^{1/2}-\tau\si^{-1/2}v -\sqrt2\, \r+\textrm{lower order terms}
\\
= -\si^{-1/6}v+o(\si^{-1/6}),
\end{multline*}
because, by the very definition of $\si$ and $\tau$, 
$$
2\si^{1/2}-\sqrt2\, \r=0 \quad \textrm{and} \quad \tau\si^{-1/2}=\si^{-1/6}.
$$

$\bullet$ The coefficient of $g'(v)$ is
$$
-(\sqrt{\si-\tau v+1}-\sqrt{\si-\tau v})\tau^{-1}=o(\si^{-1/6}).
$$

$\bullet$ The coefficient of $g''(v)$ is equal to
$$ 
\frac12(\sqrt{\si-\tau v+1}+\sqrt{\si-\tau v})\tau^{-2}=\si^{-1/6}+o(\si^{-1/6}).
$$

Hence, after multiplication by $c=\si^{1/6}$ we obtain in the limit 
$$
-v g(v)+g''(v)=D^\Airy g(v),
$$ 
as desired.
\end{proof}

\subsection{Discrete Laguerre ensemble $\to$ Airy ensemble}\label{sc:7.3}
We state the result for the plus and minus versions of the discrete Laguerre simultaneously. The initial large parameters are $\r>0$ and $\be>0$. We assume that $\r$ and $\be$ have the same order of growth, but $\r>\be$ in the case of $\DL^+(\r;\be)$, and $\r<\be$ in the case of $\DL^-(\r;\be)$. More precisely, there exists $\epsi>0$ such that 
\begin{gather*}
1+\epsi<\frac\r\be<\frac1\epsi \quad\textrm{in the case of $\DL^+(\r;\be)$},\\
1+\epsi<\frac\be\r<\frac1\epsi \quad\textrm{in the case of $\DL^-(\r;\be)$}.
\end{gather*}
Keeping this in mind, the limit regime is specified by setting
\begin{equation}\label{eq7.G}
\si=\si(\r,\be):=\frac{(\r-\be)^2}{4\r}, \quad \tau=\tau(\r,\be):=\frac{|\r^2-\be^2|^{2/3}}{16^{1/3}\r}.
\end{equation}
This entails, in particular, that $\si\asymp\r$ and $\tau\asymp \si^{1/3}$.

We are going to show that in this limit regime,
$$
\frac{\si-\DL^\pm(\r;\be)}{\tau}\to\Airy.
$$
The argument is heuristic, the exact assertion is Proposition \ref{prop7.B} below.

Let us denote by $D^{\DL^\pm(\r;\be)}$ the difference operator on $\Z_{\ge0}$ defined by the Jacobi matrix associated with $\DL^\pm(\r;\be)$. The action of $D^{\DL^\pm(\r;\be)}$ on a test function $f$ on $\Z_{\ge0}$ is given by
\begin{multline}\label{eq7.B}
(D^{\DL^\pm(\r;\be)}f)(x)=\sqrt{(x+1)(x+\be)}f(x+1)\\
\pm(2x+\be-\r)f(x)+\sqrt{x(x+\be-1)}f(x-1).
\end{multline}

\begin{proposition}\label{prop7.B}\label{pr:7.2}
Fix an arbitrary smooth function $g(v)$ on $\R$ and assign to it a function $f(x)$ on $\Z_{\ge0}$ by setting 
$$
f(x)=g(v), \quad v=\frac{\si-x}\tau.
$$
Let $\si$ and $\tau$ depend on $\r$ as in \eqref{eq7.G}, and set 
$$
c=c(r):=\sqrt{\si(\si+\be)}\tau^{-2}.
$$
As $\r\to+\infty$, we have
$$
(c(r))^{-1} D^{\DL^\pm(\r;\be)}f(x) \to D^\Airy g(v).
$$
\end{proposition}

We need a lemma.

\begin{lemma}
The parameters $\si$ and $\tau$ as defined above satisfy the system of equations 
\begin{gather}
2\sqrt{\si(\si+\be)}\pm(2\si+\be-\r)=0, \label{eq7.C}\\
\sqrt{\si(\si+\be)}\tau^{-2}=\tau\left(\frac{2\si+\be\pm2\sqrt{\si(\si+\be)}}{\sqrt{\si(\si+\be)}}\right).\label{eq7.D}
\end{gather}
\end{lemma}

Here and below the upper sign is taken for $\DL^+(\r;\be)$ and the lower sign is taken for $\DL^-(\r;\be)$. The origin of this system will be clear from the proof of Proposition \ref{prop7.B}. 

\begin{proof}
Let us check \eqref{eq7.C}. We have 
\begin{equation}\label{eq7.E}
\si=\dfrac{(\r-\be)^2}{4\r} \quad \text{and} \quad \si+\be=\dfrac{(r+\be)^2}{4r},
\end{equation}
whence 
$$
\sqrt{\si(\si+\be)}=\frac{|\r-\be|(\r+\be)}{4\r}=\pm\frac{(\r-\be)(\r+\be)}{4\r}.
$$
It follows that
$$
2\sqrt{\si(\si+\be)}\pm(2\si+\be-\r)=\pm\left(\frac{(\r-\be)(\r+\be)}{2r}+\frac{(\r-\be)^2}{2\r}+\be-\r\right)=0,
$$ 
as desired.

Let us check \eqref{eq7.D}. This equation is equivalent to
$$
\tau^3=\frac{\si(\si+\be)}{2\si+\be\pm2\sqrt{\si(\si+\be)}},
$$
while the definition of $\tau$ says that
$$
\tau^3=\frac{(\r^2-\be^2)^2}{16\r^3}.
$$
Using \eqref{eq7.E} it is readily seen that the both expressions are the same.
\end{proof}

\begin{proof}[Proof of Proposition \ref{prop7.B}]
We argue as in the proof of Proposition \ref{prop7.A}. Keeping the first three terms of the Taylor expansion of $f(x\pm1)=g(v\mp \tau^{-1})$ we write the right-hand side of \eqref{eq7.B} in the form
\begin{gather*}
(\sqrt{(x+1)(x+\be)}+\sqrt{x(x+\be-1)}\pm(2x+\be-\r))g(v)\\
-(\sqrt{(x+1)(x+\be)}-\sqrt{x(x+\be-1)})\tau^{-1}g'(v)\\
+\frac12(\sqrt{(x+1)(x+\be)}+\sqrt{x(x+\be-1)})\tau^{-2}g''(v)+\dots
\end{gather*}
Now we substitute $x=\si-\tau v$ and compute the coefficients of $g(v)$, $g'(v)$, and $g''(v)$ up to negligible terms.

$\bullet$ The coefficient of $g(v)$ is equal to
\begin{multline*}
\sqrt{(\si-\tau v+1)(\si-\tau v+\be)}+\sqrt{(\si-\tau v)(\si-\tau v+\be-1)}\pm(2\si-2\tau v+\be-\r)\\
=\sqrt{\si(\si+\be)}\left\{\sqrt{\left(1-\frac{\tau v-1}\si\right)\left(1-\frac{\tau v}{\si+\be}\right)}+
\sqrt{\left(1-\frac{\tau v}\si\right)\left(1-\frac{\tau v+1}{\si+\be}\right)}\right\}\\
\pm(2\si-2\tau v+\be-\r).
\end{multline*}
The expression in the curly brackets is equal to
$$
2\left(1-\frac{\tau v}{2\si}-\frac{\tau v}{2(\si+\be)}+\ldots\right)=2-\frac{(2\si+\be)\tau v}{\si(\si+\be)}+\ldots\,.
$$
Hence, the whole expression can be written as
$$
\left[2\sqrt{\si(\si+\be)}\pm(2\si+\be-\r)\right]-\tau\left(\frac{2\si+\be\pm2\sqrt{\si(\si+\be)}}{\sqrt{\si(\si+\be)}}\right)v+\ldots\,.
$$

By virtue of \eqref{eq7.C},  the expression in the square brackets vanishes.
We conclude that the expression for the coefficient of $g(v)$ has the form 
$$
-\tau\left(\frac{2\si+\be\pm2\sqrt{\si(\si+\be)}}{\sqrt{\si(\si+\be)}}\right)v+\ldots\,.
$$

$\bullet$ As in the proof of Proposition \ref{prop7.A}, we verify that the coefficient of $g'(v)$ is negligable.

$\bullet$ The coefficient of $g''(v)$ has the form
$$
\sqrt{\si(\si+\be)}\tau^{-2}+\ldots\,.
$$

Finally, by virtue of \eqref{eq7.D} and the definition of $c$, we see that dividing by $c$ and passing to the limit we obtain the Airy differential operator. 
\end{proof}

\section{Stochastic higher spin six vertex model and Schur measures}

\subsection{The stochastic higher spin six vertex model}\label{sc:6v}
Our exposition in this section largely follows \cite{BP-hom}, see also \cite{B-6v}. We only consider a homogeneous version of the model, for the fully general inhomogeneous version see \cite{BP-inhom}.

Consider an ensemble $\EuScript P$ of infinite oriented up-right paths drawn in the first quadrant $\Z_{\ge 1}^2$ of the square lattice, with all the paths starting from a left-to-right arrow entering each of the points $\{(1,m):m\in\Z_{\ge 1}\}$ on the left boundary (no path enters through the bottom boundary). Assume that no two paths share any horizontal piece (but common vertices and vertical pieces are allowed). See Figure \ref{fig:intro}.

\begin{figure}[htb]
	\begin{tikzpicture}
		[scale=.7,thick]
		\def\d{.1}
		\foreach \xxx in {1,...,6}
		{
		\draw[dotted, opacity=.4] (\xxx-1,5.5)--++(0,-5);
		\node[below] at (\xxx-1,.5) {$\xxx$};
		}
		\foreach \xxx in {1,2,3,4,5}
		{
		\draw[dotted, opacity=.4] (0,\xxx)--++(5.5,0);
		\node[left] at (-1,\xxx) {$\xxx$};
		\draw[->, line width=1.7pt] (-1,\xxx)--++(.5,0);
		}
		\draw[->, line width=1.7pt] (-1,5)--++(1-3*\d,0)--++(\d,\d)--++(0,1-\d);
		\draw[->, line width=1.7pt] (-1,4)--++(1-2*\d,0)--++(\d,\d)--++(0,1-2*\d)--++(\d,2*\d)--++(0,1-\d);
		\draw[->, line width=1.7pt] (-1,3)--++(1-\d,0)--++(\d,\d)--++(0,1-2*\d)--++(\d,2*\d)--++(0,1-2*\d)--++(\d,2*\d)--++(0,1-\d);
		\draw[->, line width=1.7pt] (-1,2)--++(1,0)--++(0,1-\d)--++(\d,\d)--++(1-\d,0)
		--++(0,1)--++(2-\d,0)--++(\d,\d)--++(0,2-\d);
		\draw[->, line width=1.7pt]
		(-1,1)--++(3,0)--++(0,2)--++(1,0)--++(0,1-\d)--++(\d,\d)--++(1-\d,0)
		--++(0,2);
		\draw[densely dashed] (-.5,4.5)--++(4,-4) node[above,anchor=west,yshift=16,xshift=-9] {$x+y=5$};
	\end{tikzpicture}
	\caption{A path collection $\EuScript P$.}
	\label{fig:intro}
\end{figure}

Define a probability measure on the set of such path ensembles in the following Markovian way. For any $n\ge 2$, assume that we already have a probability distribution on the intersection $\EuScript P_n$ of $\EuScript P$ with the triangle $T_n=\{(x,y)\in \Z_{\ge 1}^2: x+y\le n\}$. We are going to increase $n$ by 1. For each point $(x,y)$ on the upper boundary of $T_n$, i.e., for $x+y=n$, every $\EuScript P_n$ supplies us with two inputs: (1) The number of paths that enter $(x,y)$ from the bottom --- denote it by $i_1\in\Z_{\ge 0}$; (2) The number of paths that enter $(x,y)$ from the left --- denote it $j_1\in\{0,1\}$. Now choose, independently for all $(x,y)$ on the upper boundary of $T_n$, the number of paths $i_2$ that leave $(x,y)$ in the upward direction, and the number of paths $j_2$ that leave $(x,y)$ in the rightward direction, using the probability distribution with weights
of the transitions $(i_1,j_1)\to (i_2,j_2)$ given by
\begin{align}\label{intro-weights}
	\begin{array}{rclrcl}
		\Prob((i_1,0)\to (i_2,0))=&\dfrac{1-q^{i_1} s u}{1-su}\,\mathbf 1_{i_1=i_2},\\
		\rule{0pt}{20pt}
		\Prob((i_1,0)\to (i_2,1))=&\dfrac{(q^{i_1}-1)su}{1-su}\,\mathbf 1_{i_1=i_2+1},\\
		\rule{0pt}{20pt}
		\Prob((i_1,1)\to (i_2,1))=&\dfrac{s^2q^{i_1}-su}{1-su}\,\,\mathbf 1_{i_1=i_2},\\ 
		\rule{0pt}{20pt}
		\Prob((i_1,1)\to (i_2,0))=&\dfrac{1-s^2q^{i_1}}{1-su}\,\mathbf 1_{i_1=i_2-1}.
	\end{array}
\end{align}

Assuming that all above expressions are nonnegative, this procedure defines a probability measure on the set of all $\EuScript P$'s because we always have $\sum_{i_2,j_2} \Prob((i_1,j_1)\to (i_2,j_2))=1$, and $\Prob((i_1,j_1)\to (i_2,j_2))$ vanishes unless $i_1+j_1=i_2+j_2$.

In what follows we will use the following assumptions on the parameters:
$$
0<q<1,\qquad 
u>0,\qquad 
\text{either  } s=q^{-1/2} \text{  or  } s=-q^{1/2}.
$$

If $s=q^{-1/2}$ then $\Prob((1,1)\to (2,0))=0$, which means that no two paths can share the same vertical segment; this is the case of the stochastic six vertex model introduced in \cite{GwaSpohn} and recently studied in \cite{BCG}, see also \cite{CP}. In this case we impose a stronger condition $u>q^{1/2}$. 

It is easy to see that our assumptions on the parameters guarantee the positivity of the weights \eqref{intro-weights}. 

Each path ensemble $\EuScript P$ can be encoded by a \emph{height function}
$\mathfrak h:\Z_{\ge 1}\times\Z_{\ge 1} \to \Z_{\ge 0}$, that assigns to each 
vertex $(M,N)$ the number $\mathfrak h(M,N)$ of paths in $\EuScript P$ that pass 
through or to the right of this vertex.
The value $\HT(M,N)$ clearly depends only on the behavior of the paths in the 
$(M-1) \times N$ rectangle formed by $(M-1)$ first columns and $N$ first rows. 

\subsection{The Schur measures}
Our notation for partitions, symmetric functions, etc. below is mostly the standard one used in \cite{Macdonald1995}. 

Let $\Y$ be the set of all partitions and $\Sym$ be the algebra of symmetric functions in indeterminates $x_1,x_2,\dots$. A particularly nice linear basis of $\Sym$ is formed by the Schur symmetric functions $s_\la(x_1,x_2,\dots)$ indexed by $\la\in \Y$. 

The Schur symmetric polynomials are defined as the restriction of $s_{\lambda}$'s to a finite number of variables $x_1,\ldots, x_m$ and written as $s_{\lambda}(x_1,\ldots, x_m)$. If $m<\ell(\lambda)$ then $s_{\lambda}(x_1,\ldots, x_m)=0$.

For any two sets of indeterminates $x_1,x_2,\ldots$ and $y_1,y_2,\ldots$ define
\begin{equation}\label{PiDef1}
\Pi(x;y)=\sum_{\lambda\in \Y} s_{\lambda}(x) s_{\lambda}(y), \qquad \Pi'(x;y)=\sum_{\lambda\in \Y} s_{\lambda}(x) s_{\lambda'}(y),
\end{equation}
where $\la'$ stand for a partition whose Young diagram is conjugate to that of $\la$.  

The well-known Cauchy and dual Cauchy identities for the Schur functions read
\begin{equation}\label{eqn12}
\Pi(x;y)= \prod_{i,j} \frac{1}{1-x_iy_j},\qquad \Pi'(x;y)= \prod_{i,j} (1+x_iy_j). 
\end{equation}

The branching rule for the Schur symmetric functions \cite[Section I.5]{Macdonald1995} immediately implies that $s_\la(x_1,x_2,\dots)\ge 0$ for any $x_i\ge 0$ with $\sum_i x_i<\infty$. 

The following definition first appeared in \cite{Ok-schur}. 

\begin{definition}\label{macprocessdef}
For any two sets of nonnegative variables $x=(x_1,x_2,\dots)$ and $y=(y_1,y_2,\dots)$ such that $\Pi(x;y)<\infty$ (or $\Pi'(x;y)<\infty$) define the corresponding Schur measure as the probability measure on $\Y$ that assigns to a partition $\la\in\Y$ the weight
\begin{equation*}
\SM(x;y)(\la)=\frac{s_\la(x)s_\la(y)}
{\Pi(x;y)}\qquad \text{or}\qquad \SM'(x;y)(\la)=\frac{s_\la(x)s_{\la'}(y)}
{\Pi'(x;y)}\,.
\end{equation*}
\end{definition}

Let us now establish a correspondence between certain instances of the Schur measures $\SM$, $\SM'$ and classical orthogonal polynomial ensembles. 

For any $m,n\ge 1$, denote by $\Y^{(m)}$ the set of partitions with no more than $m$ nonzero parts: $$\Y^{(m)}=\{\la\in\Y: \ell(\la)\le m\},$$ and by $\Y^{(m,n)}$ the subset of $\Y^{(m)}$ of partitions with the largest part $\le n$: $$\Y^{(m,n)}=\{\la\in\Y: \ell(\la)\le m,\ \la_1\le n\}.$$ 

The importance of these sets to us is explained by the fact that $\SM(x,y)$ is supported by $\Y^{(m)}$ if the set of $x$-variables, or the set of $y$-variables, contains no more than $m$ nonzero elements. Similarly, $\SM'(x,y)$ is supported by $\Y^{(m,n)}$ if there are no more than $m$ nonzero $x$'s and no more than $n$ nonzero $y$'s. This is simply because the Schur function $s_\la(x_1,x_2,\dots)$ vanishes if $\ell(\la)$ is greater than the number of nonzero $x_j$'s.

\begin{proposition}\label{pr:schur-meixner} Fix $a,b\in\Z_{\ge 1}$ and $x,y\in \R_{>0}$ with $xy<1$. Consider the Schur measure $\SM\bigl(\underbrace{x,\dots,x}_{a};\underbrace{y,\dots,y}_{b}\bigr)$; it is supported by $\Y^{(\min(a,b))}$. 
The pushforward of this measure under 
\begin{equation}\label{eq:schur-meixner}
\Y^{(\min(a,b))}\to \Conf(\Z_{\ge 0}),\qquad \la \mapsto \{\min(a,b)+\la_i-i\}_{i=1}^{\min(a,b)},
\end{equation}
coincides with the $\min(a,b)$-point Meixner ensemble $\Meixner(\min(a,b),|a-b|+1,xy)$.
\end{proposition} 
\begin{proof} A straightforward computation based on Weyl's dimension formula
\begin{equation}\label{eq:weyl-dim}
s_\la(\underbrace{x,\dots,x}_a)=x^{|\la|}\prod_{1\le i<j\le a}\frac{\la_i-i-\la_j+j}{j-i}\,.
\end{equation}
\end{proof}
\begin{proposition}\label{pr:schur-krawtchouk} Fix $a,b\in\Z_{\ge 1}$ and $x,y\in \R_{>0}$. The the pushforward of the Schur measure $\SM'\bigl(\underbrace{x,\dots,x}_{a};\underbrace{y,\dots,y}_{b}\bigr)$ (which is supported by $\Y^{(a,b)}$) under the map 
\begin{equation}\label{eq:schur-krawtchouk}
\Y^{(a,b)}\to \Conf(\{0,1,\dots,a+b-1\}),\qquad \la \mapsto \{a+\la_i-i\}_{i=1}^{\min(a,b)},
\end{equation}
coincides with the $a$-point Krawtchouk ensemble 
$\Krawtchouk(a,xy/(1+xy),a+b-1)$.
\end{proposition} 
\begin{proof} Another straightforward computation with \eqref{eq:weyl-dim}. \end{proof}

\subsection{Matching expectations} In \cite{B-6v} it was shown that averages of certain observables on the higher spin six vertex model are equal to other averages over the Schur (more generally, Macdonald) measures. Let us restate here the results we will need. 

\begin{proposition}\label{pr:match}\cite[Example 4.3, Corollary 4.4]{B-6v} Take any $0<q<1$, $M,N\ge 1$, and $u>0$. Then for $s=q^{-1/2}$ and any $\zeta\notin -q^{\Z_{\le 0}}$ we have
$$
\E_{\vv} \prod_{i\ge 0} \frac{1}{1+\zeta q^{\HT(M,N)+i}}  =\E_{\SM} \prod_{j\ge 0} \frac{1+\zeta q^{\la_{N-j}+j}}{1+\zeta q^j}\,,
$$ 
where in the right-hand side we assume that $q^{\la_{-m}}=0$ for $m\ge 0$, and the Schur measure expectation is over $\SM\bigl(\underbrace{u^{-1},\dots,u^{-1}}_{N};\underbrace{q^{-1/2},\dots,q^{-1/2}}_{M-1}\bigr)$. On the other hand, for $s=-q^{1/2}$ we have
$$
\E_{\vv} \prod_{i\ge 0} \frac{1}{1+\zeta q^{\HT(M,N)+i}}  =\E_{\SM'} \prod_{j\ge 0} \frac{1+\zeta q^{\la_{N-j}+j}}{1+\zeta q^j}\,,
$$ 
where the expectation in the right-hand side is over $\SM'\bigl(\underbrace{u^{-1},\dots,u^{-1}}_{N};\underbrace{q^{-1/2},\dots,q^{-1/2}}_{M-1}\bigr)$.
\end{proposition}

Let us now use Propositions \ref{pr:schur-meixner} and \ref{pr:schur-krawtchouk} to rewrite the right-hand sides of the above identities in terms of the orthogonal polynomial ensembles. 

For $q\in [0,1)$, $\zeta\in \C\setminus\{-q^{\Z_{\le 0}}\}$, and a point configuration $X\in \Conf(\Z_{\ge 0})$, define 
$$
\mathfrak L^{(q)}_X(\zeta)=\prod_{x\in X} \frac 1{1+\zeta q^x}\,. 
$$
Note that this is a bounded continuous function on $\Conf(\Z_{\ge 0})$.

If $X$ is distributed according to a random point process $\mathbb P$ on $\X\subset \Z_{\ge 0}$, we define
$$
\mathfrak L^{(q)}_{\mathbb P}(\zeta)=\E_{X\in\Conf(\X)}\prod_{x\in X} \frac 1{1+\zeta q^x}\,. 
$$
\begin{corollary}\label{cr:6v-poly} In the notation of Proposition \ref{pr:match}, in the case $s=q^{-1/2}$, assuming $M>N$ we have
\begin{equation}\label{eq:6v-Meixner}
\E_{\vv} \prod_{i\ge 0} \frac{1}{1+\zeta q^{\HT(M,N)+i}} = \mathfrak  L^{(q)}_{\Meixner^\circ\left(N,M-N,q^{-\frac 12}u^{-1}\right)}(\zeta), 
\end{equation}
and in the case $s=-q^{1/2}$ we have
\begin{equation}\label{eq:6v-Krawtchouk}
\E_{\vv} \prod_{i\ge 0} \frac{1}{1+\zeta q^{\HT(M,N)+i}} = \mathfrak L^{(q)}_{\Krawtchouk^\circ\left(N,(1+q^{\frac12}u)^{-1},M+N-2\right)}(\zeta),
\end{equation}
where, as before, $\mathbb P\mapsto \mathbb P^\circ$ stands for the particle/hole involution of  $\mathbb P$ viewed as a point process on $\Z_{\ge0}$. 
\end{corollary}
\begin{proof} This is a direct substitution of Propositions \ref{pr:schur-meixner} and \ref{pr:schur-krawtchouk} into Proposition \ref{pr:match}, where one uses the obvious fact that for a partition $\la\in\Y^{(N)}$, $\{\la_{N-j}+j\}_{j=1}^N=\{N+\la_i-i\}_{i=1}^N$.
\end{proof}
\begin{remark}\label{rm:Meixner-shift} The relation \eqref{eq:6v-Meixner} can be easily extended to the case $M\le N$ as well. The extra caveat is that now $\min(M-1,N)=M-1$, the corresponding Schur measure is supported by $\Y^{(M-1)}$, and 
$$
\{\la_{N-j}+j\}_{j=1}^N=\{0,1,\dots,N-M\}\sqcup\{(N-M)+M+\la_i-i\}_{i=1}^{M-1}.
$$
Thus, we obtain 
\begin{equation}\label{eq:6v-Meixner-shift}
\E_{\vv} \prod_{i\ge 0} \frac{1}{1+\zeta q^{\HT(M,N)+i}} =\mathfrak  L^{(q)}_{N-(M-1)+\Meixner^\circ\left(M-1,N-M+2,q^{-\frac 12}u^{-1}\right)}(\zeta),
\end{equation}
where we used the notation $S+\mathbb P$ to denote the random point process obtained from $\mathbb P$ by the deterministic shift of all points of the random point configuration $X$ by $S$: $X\mapsto X+S$. In the case $M-1=N$ the two relations \eqref{eq:6v-Meixner} and \eqref{eq:6v-Meixner-shift} coincide, as it should be. 

The need for the deterministic shift in \eqref{eq:6v-Meixner-shift} can also be explained by the fact that $\HT(M,N)\ge N-(M-1)$, and thus we cannot possibly have the equality of the form \eqref{eq:6v-Meixner-shift} unless all the particles of the point process in the right-hand side are located strictly to the right of $N-M$.  
\end{remark}

\begin{remark} It is not hard to show, see \cite{B-duality}, that 
\begin{multline*}
\Krawtchouk^\circ\left(N,\left(1+q^{\frac 12}u\right)^{-1},\, M+N-2\right)\\ =\Krawtchouk\left(M-1,\left(1+q^{-\frac 12}u^{-1}\right)^{-1},\, M+N-2\right).
\end{multline*}
Thus, \eqref{eq:6v-Krawtchouk} can be rewritten as
$$
\E_{\vv} \prod_{i\ge 0} \frac{1}{1+\zeta q^{\HT(M,N)+i}} = \mathfrak L^{(q)}_{\Krawtchouk\left(M-1,(1+q^{-\frac 12}u^{-1})^{-1},M+N-2\right)}(\zeta).
$$
\end{remark} 

\section{Probabilistic lemmas} In this section we collect a few simple probabilistic statements that we will need later on. 

We will only deal with random variables that take values in $\Z_{\ge 0}$. For any such random variable $\xi$ and any $0<q<1$ we define
$$
\LL_{\xi}(\zeta)=\E \prod_{i\ge 1} \frac 1{1+\zeta q^{\xi+i}}, \qquad \zeta\notin -q^{\Z_{\le 0}}.
$$
This is obviously a meromorphic function of $\zeta$ with possible poles in $-q^{\Z_{\le 0}}$. 

Using the $q$-binomial theorem and the fact that $0< q^\xi \le 1$, for $\zeta\in \C$ with $|\zeta|<1$, we have
\begin{equation}\label{eq:L-exp}
\LL_{\xi}(\zeta)=\E\sum_{n\ge 0}\frac{(-\zeta)^n q^{n\xi}}{(q;q)_n}=\sum_{n\ge 0}\frac{(-\zeta)^n\E \left(q^{n\xi}\right)}{(q;q)_n}\,.
\end{equation}
As $\lim_{q\to 1} (q;q)_n/(1-q)^n=n!$, it is natural to view $\LL$ as a $q$-analog of the Laplace transform. 
The inverse transform is provided by the following statement:

\begin{lemma}\cite[Proposition 3.1.1]{BC} One may recover the probability distribution of a $\Z_{\ge 0}$-valued random variable $\xi$ from $\LL_\xi$ as follows:
\begin{equation}\label{eq:inversion}
\Prob\{\xi=n\}=\frac{-q^n}{2\pi\i}\oint_{C_n} (-q^{n+1}z;q)_\infty \LL(z) dz,
\end{equation}
where $C_n$ is any positively oriented contour which encircles the poles $z=-q^{-m}$ for $0\le m\le n$. 
\end{lemma}

Note that one can also recover the moments of $q^\xi$ directly from \eqref{eq:L-exp}:
\begin{equation}\label{eq:moments}
\E q^{n\xi}=(-1)^n(q;q)_n\oint_{|z|=const<1} \LL(z)\frac{dz}{z^{n+1}}\,, \qquad n\ge 0. 
\end{equation}

\begin{lemma} Let $\{\xi_m\}_{m\ge 1}$ be a sequence of $\Z_{\ge 0}$-valued random variables, and 
assume that there is a random variable variable $\xi$ such that $\lim_{m\to\infty}\xi_m=\xi$ in distribution. Then $\LL_{\xi_m}(\zeta)\to \LL_\xi(\zeta)$ as $m\to\infty$, uniformly in $\zeta$ varying over any compact subset of the open unit disc.  
\end{lemma}
\begin{proof} The convergence in distribution implies the convergence of the $q$-moments $\E q^{n\xi_m}\to\E q^{n\xi}$ for any $n\ge 1$ (because $q^k\le 1$ for $q<1$ and $k\ge 0$), and this implies the convergence of the right-hand sides of \eqref{eq:L-exp} uniformly over compact sets in the unit disc, where the series are term-wise majorated by 
$$
\sum_{n\ge 0}\frac{r^n}{(q;q)_n}, \qquad r<1.  
$$
\end{proof}

A converse statement is also true. 

\begin{lemma}\label{lm:qL-limit} Let $\{\xi_m\}_{m\ge 0}$ be a sequence of $\Z_{\ge 0}$-valued random variables, and 
assume that the sequence of the $q$-Laplace transforms $\bigl\{F_m(\zeta)=\LL_{\xi_m}(\zeta)\bigr\}_{m\ge 1}$ converges to a function $F(\zeta)$ uniformly in $\zeta$ varying over any compact subset of the open unit disc. Then there exists a random variable $\xi$ such that $\xi_m\to\xi$ as $m\to\infty$ in distribution, and $F(\zeta)=\LL_\xi(\zeta)$. 
\end{lemma}
\begin{proof} The uniform convergence of holomorphic functions in a domain implies the convergence of all their derivatives at every point of the domain. On the other hand, the convergence of the derivatives at a point implies, by Taylor expansion, the uniform convergence of the functions in compact subsets of any open disc centered at this point, if all the functions are holomorphic in this disc. Iterating these statements implies the uniform convergence of $\{F_m\}_{m\ge 1}$ in any compact subset of $\C\setminus\bigl\{-q^{\Z_{\le 0}}\bigr\}$, which implies the convergence of \eqref{eq:inversion} with $\xi_m$ instead of $\xi$ and $m\to\infty$. Set $p_k=\lim_{m\to\infty} \Prob\{\xi_m=k\}$, $k\ge 0$.  
Clearly, $p_k\ge 0$ and $\sum_{k\ge 0} p_k\le 1$. We thus have for any $n\ge 0$
$$
\lim_{m\to\infty} \E q^{n\xi_m}=\sum_{k\ge 0} p_k q^{nk},\qquad \lim_{m\to\infty} \E q^{n\xi_m}=(-1)^n(q;q)_n\oint_{|z|=const<1} F(z)\frac{dz}{z^{n+1}},
$$
where we used \eqref{eq:moments} for the second relation. Substituting $n=0$ yields $\sum_{k\ge 0} p_k=F(0)=\lim_{m\to\infty} F_m(0)=1$, which means that we can define a random variable $\xi$ by $\Prob\{\xi=k\}=p_k$. Finally, the above limiting relations imply, via \eqref{eq:L-exp}, that $\LL_\xi=F$.  
\end{proof}

\section{From the six vertex model to the ASEP to discrete Laguerre ensembles} The asymmetric simple exclusion process, or ASEP for short, is a well-known interacting particle system on (an interval of) the one-dimensional lattice that has been extensively studied since its introduction in \cite{Spitzer} and \cite{MGP}. The system consists of particles occupying vertices of $\Z$, no more than one per site, that randomly move in continuous time. The evolution is Markovian, and can be informally described as follows: Each particle has two exponential clocks of rates $\l$ and $\mathfrak r$, call them left and right. When the left clock rings, the particle moves to the left by one if the corresponding target site is unoccupied, and nothing happens if that site is occupied (the jump is suppressed). Similarly, if the right clock rings, the particle moves to the right by one if the corresponding site is empty. All the clocks are independent. A more detailed description, as well as a proper definition of this dynamics with infinitely many particles, can be found in \cite{Liggett}. 

We next observe that the stochastic six vertex model as described in Section \ref{sc:6v} can be used to approximate the ASEP with a particular \emph{packed} or \emph{step} initial condition, when at time 0 the particles occupy all negative integers. This is explained in \cite[Section 6.5]{BP-hom} in detail and proved in \cite{Aggarwal}, but the idea is fairly simple; let us describe it. One reads the path ensemble $\mathcal P$ row by row, thinking of the places, where the paths intersect a given horizontal section, as of particle locations. Moving the horizontal section upward will correspond to increasing the time. Further, one tunes the parameters in such a way that the probability for a path not to turn is infinitesimally small, i.e. 
$\Prob\{1,0;1,0\}$ and $\Prob\{0,1;0,1\}$ are both small. This leads to paths becoming staircase-like, with all the steps having height and width one. If we now measure the particle positions using a moving frame with the $\Z$-origin in each row corresponding to the points $(n,n+1)$ in the quadrant, we will observe particles accumulating at the negative integer locations as our horizontal sections move higher. Finally, taking $\Prob\{1,0;1,0\} \sim \epsilon \l$, $\Prob\{0,1;0,1\}\sim \epsilon \mathfrak r$, and scaling the height of the horizontal section as $\epsilon^{-1} t$, $\epsilon\ll 1$, we see the convergence to the ASEP, as the appearance of vertices of the type $\{0,1;0,1\}$ means that the corresponding particle jumps to the right, and, similarly, vertices of the type $\{1,0;1,0\}$ give jumps to the left. 

Let us now make a precise statement. For ASEP particle configurations, we introduce a height function $\HT:\Z\to\Z$ that counts the number of particles weakly to the right of a given location. Since the particle configurations are random, this is a random function. We will only consider the situation when ASEP particles do not accumulate at $+\infty$, so that the height function is always finite.

\begin{proposition}\label{pr:6v-ASEP} Consider the stochastic six vertex model in the quadrant with 
$$
s=q^{-1/2},\qquad\qquad u=q^{-1/2}+(1-q)q^{-1/2}\epsilon,\quad \epsilon>0.
$$ 
Also consider the ASEP on $\Z$ with particles occupying all negative integers at time 0, and with the jump rates
$\l=q$, $\mathfrak r=1$. Then for any $x\in\Z,\ \zeta\in\C\setminus\{-q^{\Z_{\le 0}}\},\ t\ge 0$, we have
\begin{equation}\label{eq:6v-ASEP}
\lim_{\epsilon\to 0} \LL_{\HT(M,N)}(\zeta)=\LL_{\HT(x)}(\zeta),\qquad  M=[\epsilon^{-1}t]+x+1,\quad N=[\epsilon^{-1}t],
\end{equation}
where the convergence is uniform on compact sets in the open unit disc, and on the left\/ $\HT$ stands for the height function of the six vertex model, while on the right\/ $\HT$ stands for the height function of the ASEP at time $t$. 
\end{proposition} 
\begin{proof} The $q$-moments of the height function of the stochastic six vertex model in the quadrant were computed in \cite[Theorem 4.12]{BCG}, while the $q$-moments of the ASEP height function with the step initial condition were computed in \cite[Theorem 4.20]{BCS14}. The fact that one formula converges to the other in the limit regime of the proposition is essentially obvious, the only needed limiting relation can be found in \cite[Corollary 10.2]{BP-hom}. On the other hand, convergence of the $q$-moments implies the convergence of the $q$-Laplace transforms $\LL$ via \eqref{eq:L-exp}. 
\end{proof} 

We can now establish a connection between the ASEP and discrete Laguerre ensembles. 

\begin{theorem}\label{th:ASEPDL} 
Consider the ASEP on $\Z$ with particles occupying all negative integers at time 0, and with the jump rates $\l=q\in (0,1)$, $\mathfrak r=1$, and let $\HT$ denote its height function, as above. Then at any time moment $t\ge 0$, and for any $x\in\Z,\  \zeta\in\C\setminus\{-q^{\Z_{\le 0}}\}$,  we have
\begin{equation}\label{eq:ASEP-DL}
\LL_{\HT(x)}(\zeta)= \begin{cases}\mathfrak L^{(q)}_{\DL^+((1-q)t,x+1)}(\zeta),& x\ge 0, \\
\mathfrak L^{(q)}_{-x+\DL^+((1-q)t,-x+1)}(\zeta),& x\le 0,\end{cases}
\end{equation}
where $-x+\DL^+(\,\cdot\,,\cdot\,)$ stands for the discrete Laguerre ensemble that is deterministically shifted (to the right) by $-x\ge 0$. 
\end{theorem} 
\begin{proof} This is a combination of Proposition \ref{pr:6v-ASEP}, Corollary \ref{cr:6v-poly}, and Theorem \ref{thm4.A}. 

It suffices to consider the case $x\ge 0$; the case $x\le 0$ is then obtained by the particle/hole involution that maps the ASEP into itself reflected around the origin (alternatively, one can use Remark \ref{rm:Meixner-shift}). First, $\LL_{\HT(x)}(\zeta)$ is a limit of the six vertex $\LL_{\HT(M,N)}(\zeta)$, see \eqref{eq:6v-ASEP}. Next, use \eqref{eq:6v-Meixner} to write $\LL_{\HT(M,N)}(\zeta)$ in terms of a Meixner ensemble. The weak limit of the Meixner ensemble is afforded by Theorem \ref{sc:6.2}; the limiting point process is $\DL^-((1-q)t,x+1)$ as in the notation of Proposition \ref{pr:6v-ASEP} we have
$$
M-N=x+1,\qquad\qquad  (1-q^{-1/2}u^{-1})N\to (1-q)t,\quad \epsilon\to 0. 
$$
The particle/hole involution obviously does not impact the weak convergence of the point processes on $\Z_{\ge 0}$, thus we have the weak convergence of the $\Meixner^\circ$ ensemble from \eqref{eq:6v-Meixner} to $(\DL^-((1-q)t,x+1))^\circ$. The latter process is $\DL^+((1-q)t,x+1)$, cf. Section \ref{sect3.1}. Finally for $q\in (0,1)$ and $\zeta\in\C\setminus\{-q^{\Z_{\le 0}}\}$,  $\mathfrak L^{(q)}_{X}(\zeta)$ is a bounded continuous function in $X\in\Conf(\Z_{\ge 0})$, thus the weak convergence of the point processes implies the convergence of its corresponding averages. 
\end{proof} 

Theorem \ref{th:ASEPDL} admits a limit as $q\to 0$. This turns the ASEP into the TASEP (`T' is for `totally'), and the result itself turns into \cite[Proposition 1.4]{Joh-shape} that was a crucial step for Johansson's celebrated proof of the Tracy-Widom asymptotics for the TASEP (Theorem 1.6 \emph{ibid.}). The limiting version of 
Theorem \ref{th:ASEPDL} looks as follows:

\begin{corollary}\cite[Proposition 1.4]{Joh-shape} Consider the TASEP with the unit jump rate on $\Z$, with particles occupying all negative integers at time 0. Let $\HT$ denote its height function. Then for any position $x\in\Z_{\ge 0}$ and any time moment $t\ge 0$, $\Prob\{\HT(x)\le N-1\}$ equals the probability of the event that the right-most particle in the $N$-particle Laguerre orthogonal polynomial ensemble $\Lag(N,x+1)$ is to the left of $t$.

For $x\in\Z_{<0}$, the same relation holds with $\HT(x)$ replaced by $(-x+\HT(x))$ and $\Lag(N,x+1)$ replaced by $\Lag(N,-x+1)$. 
\end{corollary}

\begin{proof} As in Theorem \ref{th:ASEPDL}, it suffices to consider $x\ge 0$.
Observe that for $y,n\in\Z_{\ge 0}$, $\zeta=q^{-n-1/2}$, $Z\in\Conf(\Z_{\ge 0})$, we have
$$
\lim_{q\to 0} \prod_{i\ge 1} \frac 1{1+\zeta q^{y+i}}=\mathbf{1}_{y>n},\qquad\qquad
\lim_{q\to 0} \prod_{z\in Z} \frac 1{1+\zeta q^z}=\mathbf{1}_{\min(Z)>n},
$$ 
and all the expressions remain bounded throughout the limit transition. 

Taking the expectations of these limit relations with $y=\HT(x)$ and $Z$ being distributed according to $\DL^+((1-q)t,x+1)$, we see that Theorem \ref{th:ASEPDL} tells us that $\HT(x)$ for the TASEP is distributed exactly as the left-most particle of $\DL^+(t,x+1)$. Employing the connection between the discrete and continuous Laguerre ensembles from Section \ref{sc:duality}, see Theorem \ref{thm3.A}, finishes the proof. 
\end{proof}

\section{The ASEP at large times: Three limit regimes}\label{sect11}

The goal of this section is to use Theorem \ref{th:ASEPDL} to analyze the behavior of the ASEP with step initial data (particles occupy all negative integers at time 0) at large times. We will prove three results, one of which corresponds to the degeneration of the discrete Laguerre ensemble to the discrete Hermite ensemble (as in Section \ref{sc:6.7}), while the other two correspond to the degeneration to the Airy ensemble (as in Section \ref{sc:7.3}). 

Let us start with the one that corresponds to the discrete Hermite ensemble; it is simpler as it does not require any scaling of the state space $\Z_{\ge 0}$ of the determinantal point process. 

\begin{proposition}\label{pr:ASEP-Hermite} Consider the ASEP on $\Z$ with particles occupying all negative integers at time 0, and with the jump rates $\l=q\in (0,1)$, $\r=1$, and let $\HT$ denote its height function at time $t\ge 0$. Then as $t\to\infty$ and for any $r\in \R$, $\HT\left({(1-q)t}-\sqrt{2(1-q)t}\cdot r \right)$ converges in distribution to a $\Z_{\ge 0}$-valued random variable $\xi_r$ characterized by 
\begin{equation}\label{eq:ASEP-Hermite}
\mathcal L_{\xi_r}^{(q)}(\zeta)=\mathfrak{L}^{(q)}_{\dHermite^+(r)}(\zeta), \qquad\qquad \zeta\in\C\setminus\{-q^{\Z_{\le 0}}\}.
\end{equation}
\end{proposition}

\begin{remark} Since the limiting values of the height function are in $\{0,1,2,\dots\}$, Proposition \ref{pr:ASEP-Hermite} describes the behavior of the finitely many first (that is, right-most) ASEP particles. An alternative description of this limiting regime was obtained by Tracy-Widom in \cite[Theorem 2]{TW3} (conjectured earlier in \cite{TW2}) in terms of Fredholm determinants. Matching our result to Tracy-Widom's one is an interesting problem, but we do not pursue it in this work.  
\end{remark}

\begin{proof}[Proof of Proposition \ref{pr:ASEP-Hermite}] 
The argument is very similar to the proof of Theorem \ref{th:ASEPDL}  above. First, by Theorem \ref{thm6.LH}, in the described limit regime the discrete Laguerre ensembles of Theorem \ref{th:ASEPDL}  weakly converges to the discrete Hermite ensemble $\dHermite^+(r)$. Hence, the averages of the bounded continuous function $\mathfrak L^{q}_{*}(\zeta)$ converge too. This shows the convergence of the right-hand side of \eqref{eq:ASEP-DL} to $\mathfrak{L}^{(q)}_{\dHermite^+(r)}(\zeta)$. It is straightforward to strengthen this convergence to the uniform one in $\zeta$ varying over compact sets in the open unit disc. (For example, one can first show that the particles outside a large enough subset $\X_M:=\{0,1,\dots,M\}\subset \Z_{\ge 0}=\X$ can only affect the values of $\mathfrak L^{q}_{\cdot}(\zeta)$ by a uniformly close to 1 multiplicative factor, and then note that there are only finitely many possible particle configurations in $\X_M$. Hence, the collection of their probabilities with respect to the discrete Laguerre ensembles converges uniformly to the corresponding probabilities of the discrete Hermite ensemble.) Finally, Lemma \ref{lm:qL-limit} implies the claim. 
\end{proof}

Now we proceed to the convergence of the discrete Laguerre ensembles to the Airy ensemble. 

\begin{theorem}\label{th:ASEP-TW} Consider the ASEP on $\Z$ with particles occupying all negative integers at time 0, and with the jump rates $\l=q\in (0,1)$, $\r=1$, and let $\HT$ denote its height function at time $t\ge 0$. Assume that $x, t\to \infty$ at the same rate, and $|x|/((1-q)t)\le const<1$. Then the random variable 
$$\frac{\sigma-\HT(x)}\tau \quad \text{   for   } \quad x\ge 0\qquad\qquad \text{      or      }\qquad\qquad \frac{\sigma-\HT(x)-x}\tau\quad \text{   for   } \quad x\le 0,$$ 
with  
$$
\sigma=\frac{(\widetilde t-|x|)^2}{4\widetilde t}\,, \qquad\qquad  \tau = \frac{({\widetilde t}^2-|x|^2)^{2/3}}{2^{4/3}\,\widetilde t}\,,\qquad\qquad \widetilde t=(1-q)t,
$$
weakly converges to the GUE Tracy-Widom distribution.  
\end{theorem}
\begin{remark} The above result is equivalent to the celebrated theorem of Tracy-Widom \cite[Theorem 3]{TW3} that says that 
$$
\frac{x_m(T/(1-q))-c_1 T}{c_2 \,T^{1/3}} \to F_{GUE}, \qquad s=\frac mT\in (0,1), \ c_1=1-2\sqrt{s}, \ c_2=s^{-1/6}(1-\sqrt{s})^{2/3},
$$
which in its turn generalized Johansson's \cite[Theorem 1.6]{Joh-shape} for $q=0$. In the above relation, $x_m(T)$ is the position of the $m$th right-most ASEP particle at time $T$ (we assume the same step initial data), and $F_{GUE}$ is the GUE Tracy-Widom distribution. 

The equivalence is established by noting that the event that at time $t$ the particle number $m$ is at position $\ge x$ is exactly the same as the event that at time $t$ and we have $\HT(x)\ge m$, and by matching the notations (for $x\ge 0$, the case $x\le 0$ is similar)
\begin{gather*}
T=\widetilde t, \quad s={\HT}/T,\\
\HT\sim\frac{T}4{\left(1-\frac{x}{T}\right)^2}-\frac{T^{1/3}}{2^{4/3}}{\left(1-\left(\frac{x}{T}\right)^2\right)^{2/3}}F_{GUE}\ \longleftrightarrow\  x\sim T(1-2\sqrt{s})+T^{1/3} \frac{(1-\sqrt{s})^{2/3}}{s^{1/6}}F_{GUE}
\end{gather*}
as $T\to\infty$, up to terms of order smaller than $T^{1/3}$. 
\end{remark}
\begin{remark} A different proof of Theorem \ref{th:ASEP-TW} for $x=0$ was given in \cite[Appendix D]{BCS14}. 
\end{remark}

\begin{proof}[Proof of Theorem \ref{th:ASEP-TW}] The argument is similar to the proof of \cite[Theorems 6.1, 6.3]{B-6v}. First, we need to recall the definition of asymptotic equivalence from \cite[Definitions 5.1 and 5.2]{B-6v}. 
With those definitions we can proceed to the proof. We will only consider the case $x\ge 0$, the case $x\le 0$ is analogous. 

Using \cite[Proposition 5.3, Example 5.5]{B-6v} we see that the family of functions
$$
F_t(y):=\mathcal L^{(q)}_{\HT(x)}(q^{y})=\mathbb E_{(\text{ASEP at time}\ t)} \prod_{i\ge 0}\frac 1{1+q^{\HT(x)+y+i}}\,\qquad t\ge 0,\ y\in \R,
$$
is asymptotically equivalent to $-\HT(x)$ as $t\to\infty$. 

On the other hand, \cite[Corollary 5.7]{B-6v} implies that the family of functions
$$
\widetilde F_t(y):=\mathfrak L^{(q)}_{\DL^+((1-q)t,x+1)}(q^y)=\mathbb{E}_{Z\in \DL^+((1-q)t,x+1) }\prod_{z\in Z} \frac1{1+q^{z+y}}\,,\qquad  t\ge 0, \ y\in \R,
$$
is asymptotically equivalent to $-\min Z$, which is the negative location of the left-most particle in $\DL^+((1-q)t,x+1)$, as $t\to\infty$. 

Since $F_t\equiv \widetilde F_t$ by Theorem \ref{th:ASEPDL}, we conclude that $\HT(x)$ is asymptotically equivalent to the position of the left-most particle in $\DL^+((1-q)t,x+1)$. The latter converges in distribution, under the scaling of Section \ref{sc:7.3} and Proposition \ref{pr:7.2}, to the distribution of the right-most particle of the Airy ensemble, which is the GUE Tracy-Widom distribution \cite{TW1}. (Proposition \ref{pr:7.2} does not actually proof the needed convergence.
However, by virtue of Theorem \ref{thm3.A}, this convergence is equivalent to the corresponding convergence of the right-most particle of the Laguerre orthogonal polynomial ensembles to the right-most particle of the Airy ensemble, which is a well-known fact that goes back to \cite{Joh-shape}.) Matching the scaling of Theorem \ref{th:ASEP-TW} and Propositon \ref{pr:7.2} completes the proof. 
\end{proof}

Before moving into a full description of the third asymptotic result, let us first explain how it can be foreseen. 

The proof of Theorem \ref{th:ASEP-TW} above essentially consists of two ingredients: (1) The convergence of the discrete Laguerre ensembles to the Airy ensemble; (2) The approximation of the observable $\mathfrak L^{(q)}_{Z}$, $Z\in \DL$, by a characteristic function of $\min Z$. The latter claim follows from the fact that each factor of the form $(1+\zeta q^z)^{-1}$ in the infinite product that defines $\mathfrak L^{(q)}$, converges to either 0 or 1. 

The new limit transition will now follow from varying the parameter $q$ that so far remained fixed in $(0,1)$. Namely, we want to send $q$ to 1 so that, after choosing an appropriate $\zeta$, each of the factors $(1+\zeta q^z)^{-1}$ would converge to a nontrivial limit in $(0,1)$ for any particle of the limiting Airy ensemble, with different limits for different particles. Since the pre-Airy particles in the discrete Laguerre ensembles have inter-particle distances of order $\widetilde t^{1/3}$, this means that we have to choose $\ln q\sim \widetilde t^{-1/3}$.

On the other hand, the $q$-Laplace transform $\mathcal L^{(q)}_{\HT(x)}$ in the limit $q\to 1$ should approximate the usual Laplace transform of the scaled height function. All this leads to the following statement. We will only consider the case $x\ge 0$, the case $x\le 0$ is completely analogous. 

\begin{theorem}\label{th:ASEP-KPZ} Consider the ASEP on $\Z$ with particles occupying all negative integers at time 0, and with the jump rates $\l=q\in (0,1)$, $\r=1$, and let $\HT$ denote its height function at time $t\ge 0$. Let $\epsilon>0$ be a small parameter, and assume that
$$
q=(1-\epsilon)\to 1, \quad t=\epsilon^{-4}\cdot \widehat t, \quad x=\epsilon^{-3}\cdot \widehat x, \qquad \widehat x/\widehat t\in [0,1). 
$$
Then the random variables
$$
\epsilon^{-2}\cdot\widehat \sigma-\ln\epsilon  -\epsilon\cdot \HT(x)\qquad\text{with}\qquad \widehat\sigma=\dfrac{(\widehat t-\widehat x)^2}{4\widehat t}
$$
have a weak limit as $\epsilon\to 0$; denote it by $\xi$. The Laplace transform of $\exp(\xi)$ is given by
\begin{equation}\label{eq:Laplace-KPZ}
\mathbb{E} \left[e^{-\widehat\zeta\exp(\xi)}\right] =\mathbb{E}_{Z\in Airy} \prod_{z\in Z} \frac{1}{1+\widehat\zeta \exp(\widehat \tau z)}\,,\qquad\qquad \widehat\zeta>0, \quad \widehat \tau=\dfrac{(\widehat t^2-\widehat x^2)^{2/3}}{2^{4/3}\widehat t}. 
\end{equation}
\end{theorem}
\begin{remark}\label{rm:ASEP-KPZ} This result (in a slightly different form) goes back to \cite{ACQ} and \cite{SS}. According to \cite{ACQ, CDR, D, SS}, the limiting random variable $\xi$ has the same distribution as $T/24-H_{KPZ}(\text{space}=0, \text{time}=T)$, where $T=2\,\widehat\tau^3$, and $H_{KPZ}$ is the Hopf-Cole solution of the KPZ (Kardar-Parisi-Zhang) stochastic partial differential equation with the so-called narrow wedge initial data. We refer to \cite{BG-KPZ} for a description of this result that is similar to the notations of the present paper.  

A general result that the \emph{weak asymmetry} limit of the ASEP ($q\to 1$) is related to solutions of the KPZ equation for certain regular initial conditions goes back to \cite{BertiniGiacomin}. 
A concrete realization for the somewhat singular narrow wedge initial data, that is equivalent to Theorem \ref{th:ASEP-KPZ}, is due to \cite{ACQ, SS}. 

A discussion on the levels of mathematical rigor of the above developments can be found in the survey \cite{C-survey}. 
\end{remark}

\begin{proof}[Sketch of the proof of Theorem \ref{th:ASEP-KPZ}] We need to take limits of the two sides of the identity
\begin{equation}\label{eq:ASEP-DLaguerre}
\E_{ASEP} \prod_{i\ge 0} \frac 1{1+\zeta q^{\HT(x)+i}} = \E_{Z\in \DL^+((1-q)t,x+1)}\prod_{z\in Z} \frac1{1+ \zeta q^z}. 
\end{equation}
Let us first explain how the limit works and then point out the steps that would turn this explanation into an actual proof. 

For the right-hand side, we use, cf. Section \ref{sc:7.3} and Proposition \ref{pr:7.2}, with $\sigma,\tau$ as in Theorem \ref{th:ASEP-TW},
$$
\DL^+((1-q)t,x+1)\sim \sigma - \tau \Airy = \epsilon^{-3}\cdot \widehat\sigma - \epsilon^{-1}\cdot  \widehat \tau \Airy. 
$$
Choosing $\zeta$ so that $\zeta q^\sigma \to \widehat \zeta$ and observing that $q^{-\epsilon^{-1}\widehat \tau z}\sim  \exp(\widehat \tau z)$ for finite $z\in \R$, we conclude that the right-hand side of \eqref{eq:ASEP-DLaguerre} should converge to that of \eqref{eq:Laplace-KPZ}. 

On the other hand, if we denote $\widehat\HT = \widehat \HT_\epsilon:=\epsilon^{-2}\cdot\widehat \sigma-\ln\epsilon  -\epsilon\cdot \HT(x)$, then 
$$
\prod_{i\ge 0} \frac1{1+\zeta q^{\HT(x)+i}}=\prod_{i\ge 0} \frac1{1+\zeta  q^{\epsilon^{-3}\cdot\widehat\sigma -\epsilon^{-1}\ln\epsilon - \epsilon^{-1}\cdot \widehat\HT+i}}
\sim \prod_{i\ge 0}  \frac1{1+\epsilon\cdot \widehat \zeta e^{\widehat \HT}\cdot q^{i}}=\sum_{k\ge 0}
\frac{(-1)^k(1-q)^k}{(q;q)_k} (\widehat\zeta e^{\widehat \HT})^k,
$$
and the last expression obviously converges to $e^{-\widehat \zeta \exp({\widehat \HT})}$ as $q\to1$ (we used the q-binomial theorem to turn the product into the sum). 

To convert this computation into a proof, we follow the following steps. 

\noindent\emph{Step 1.} We start by restricting the infinite products in the right hand-sides to include only the particles that lie (directly or after the discrete Laguerre to Airy scaling) in a subset of the form $\X_M=(-\infty,M)$ of the Airy ensemble's state space $\X=\R$. We claim that such a modification changes the observables by a multiplicative constant that is uniformly close to 1 with high probability, and that remains bounded almost surely, when we choose $M$ close enough to $-\infty$.

In the Airy case, this amounts to investigating 
$$
\prod_{z\in Z, z<M} \frac 1{1+\widehat \zeta \exp(\widehat \tau z)},\qquad Z\in\Airy. 
$$
First, this is obviously bounded with probability 1. Next, the closeness of this product to 1 can be controlled by the smallness of the additive statistic $S_M:=\sum_{z\in Z, z<M} \exp(\widehat \tau z)$. We have
\begin{gather*}
\E S_M=\int_{-\infty}^M K_\Airy (x,x)\exp(\widehat\tau x)dx, \\ 
\E S_M(S_M-1)=\int_{-\infty}^M \int_{-\infty}^M (K_\Airy (x,x)K_\Airy (y,y)-K^2_\Airy (x,y))\exp(\widehat\tau (x+y))dxdy.
\end{gather*}
Both these quantities can be made as small as one wishes by a suitable choice of $M$. Hence, by Chebyshev's inequality one can make $S_M$ arbitrarily close to 0 with arbitrarily high probability. 

The argument for the discrete Laguerre ensemble is very similar. 

\smallskip\noindent\emph{Step 2.} In the restricted range of $[M,+\infty)$, we have the trace-class convergence of the discrete Laguerre kernel to the Airy one. \footnote{This convergence does not follow from the results of the present paper, but its proof is a standard (although tedious) arguement that is based on the well-known convergence of the Laguerre polynomials to the Airy function.} This implies the convergence of the probabilities of observing a fixed number of particles in the corresponding ranges of the two ensembles, as well as the weak convergence of the distributions of the positions of those particles. Hence, the right-hand side of \eqref{eq:ASEP-DLaguerre} converges to that of \eqref{eq:Laplace-KPZ}. 

\smallskip\noindent\emph{Step 3.} The q-Laplace transform-like observables in the left-hand side of \eqref{eq:ASEP-DLaguerre} uniformly approximate the exponential observable $\exp(-\widehat \zeta\exp(\widehat \HT))$ as $\epsilon\to 0$, with $\widehat\zeta$ varying over $\R_{\ge 0}$. 

By the weak compactness of the space of positive measures, the distributions of $\HT=\HT_\epsilon$ must have limiting points in the space of positive measures on $\R_{\ge 0}$ of total mass $\le 1$; let $\mu$ be one such limiting point. Then the left-hand sides of \eqref{eq:ASEP-DLaguerre} converge to the Laplace transforms $\int_{\R_{\ge 0}} \exp(-\widehat\zeta y) \mu(dy)$ of $\mu$. Since the Laplace transforms determine such measures uniquely, $\mu$ is actually unique. Furthermore, since we already know the limit of the right-hand side of \eqref{eq:ASEP-DLaguerre},
and it has the property of approaching 1 as $\widehat \zeta\to 0$, we conclude that the total integral of $\mu$ is exactly 1, i.e. $\mu$ is a \emph{bona fide} probability measure on $\R_{\ge 0}$. Denoting by $\xi$ a random variable with distribution $\mu$ yields the desired claim.   
\end{proof}

\section{Large scale limits of the stochastic six vertex model in a quadrant}
In this section we study asymptotic regimes of the stochastic six vertex model in a quadrant as defined in Section \ref{sc:6v}, that are parallel to the asymptotic behavior of the ASEP from the previous section. 

We start with the simplest regime that is analogous to Proposition \ref{pr:ASEP-Hermite}. As in Section \ref{sect11}, let us denote by $\xi_{r}$ a $\Z_{\ge 0}$-valued random variables determined by (cf. \eqref{eq:ASEP-Hermite})
$$
\mathcal L^{(q)}_{\xi_r}(\zeta)=\mathfrak L^{(q)}_{\DH^+(r)}(\zeta),\qquad\qquad r\in\R, \quad \zeta\in \C\setminus\{-q^{\Z_{\le 0}}\}.
$$

\begin{proposition}\label{pr:6v-dHermite} Consider the {\rm(}higher spin stochastic{\rm)} six vertex model in a quadrant with $q\in(0,1)$ and $u>0$. 

{\rm(i)} For $s=q^{-1/2}$ and $u>q^{-1/2}$, assume that, for some $r\in \R$, 
$$
M,\ N\to\infty, \qquad \begin{cases} -\sqrt{q^{-1/2}u^{-1}(M-N)}+\dfrac{(1-q^{-1/2}u^{-1})N}{\sqrt{q^{-1/2}u^{-1}(M-N)}} \to \sqrt{2}\, r, &M>N,\\
-\sqrt{q^{-1/2}u^{-1}(N-M)}+\dfrac{(1-q^{-1/2}u^{-1})M}{\sqrt{q^{-1/2}u^{-1}(N-M)}} \to \sqrt{2}\, r, &M<N.
\end{cases}
$$
Then $\HT(M,N)$ for $M>N$, and $(\HT(M,N)+(M-1)-N)$ for $M<N$, converge to $\xi_r$ in distribution. 

{\rm(ii)} For $s=-q^{1/2}$, assume that $M,N\to\infty$, and
$$
-\sqrt{\frac{M+N}{1+q^{1/2}u}}+N\sqrt{\frac{1+q^{1/2}u}{M+N}}\to \sqrt{2}\, r, \qquad\qquad r\in \R. 
$$
Then $\HT(M,N)$ converges to $\xi_r$ in distribution. 
\end{proposition}
\begin{remark} The three limiting regimes of the above proposition can be achieved by introducing a large parameter $L>0$ and taking $M=\mu L+O(\sqrt{L})$, $N=\nu L+ O(\sqrt{L})$, with $\mu,\nu$ satisfying  
$$
\frac{\mu}{\nu}=
\begin{cases}
q^{1/2}u,& s=q^{-1/2},\ M>N,\quad \text{or}\quad s=-q^{1/2},\\
(q^{1/2}u)^{-1},& s=q^{-1/2}, \ M<N. 
\end{cases}
$$
\end{remark}
\begin{proof}[Proof of Proposition \ref{pr:6v-dHermite}] The argument is exactly the same as in the proof of Proposition \ref{pr:6v-dHermite}, with the use of Corollary \ref{cr:6v-poly} and the convergences of the Meixner and Krawtchouk ensembles to the discrete Hermite ensemble, see Sections \ref{sc:6.3}-\ref{sc:6.4} above. Then change of the sign in front of $\sqrt{2}r$, as compared to \eqref{eq:th6.3} and \eqref{eq:th6.4}, is explained by the fact that there is the particle-hole involution $\mathbb P \mapsto \P^\circ$ in \eqref{eq:6v-Meixner}, \eqref{eq:6v-Meixner-shift}, \eqref{eq:6v-Krawtchouk}, and $(\DH^+(-r))^\circ=\DH^+(r)$.   
\end{proof}

Proposition \ref{pr:6v-dHermite} describes the behavior of the stochastic six-vertex model near the linear boundaries that separate `liquid' and `frozen' zones. This can be seen from Theorems 6.1 and 6.3 of \cite{B-6v} that describe the asymptotic behavior in the liquid zone (Theorem 6.1 there goes back to \cite{BCG}). In fact, those two theorems are exact analogs of Theorem \ref{th:ASEP-TW} above. Their proofs relied on convergences of the Meixner and Krawtchouk ensembles to the Airy ensemble (in \cite{B-6v} the exposition is in the language of the Schur measures, the match to the Meixner and Krawtchouk ensembles is made via Propositions \ref{pr:schur-meixner} and \ref{pr:schur-krawtchouk} above). 

Very similarly to what was done with the ASEP in the previous section, one can modify the limiting argument by taking $q\to 1$ in such a way that all points of the Airy ensemble yield nontrivial factors in the averaged observable. This leads to the following statement. 

\begin{theorem}\label{th:6v-KPZ} Consider the {\rm(}higher spin stochastic{\rm)} six vertex model in a quadrant, as described in Section \ref{sc:6v}. Let $\epsilon>0$ be a small parameter, and assume that 
$$
q=(1-\epsilon)\to 0,\quad u=q^{-1/2} v^{-1},\quad  M=\epsilon^{-3}\cdot\mu, \quad N=\epsilon^{-3}\cdot\nu,
$$ 
with $(M,N)$ in the liquid zone, i.e. 
$\mu/\nu\in (v,v^{-1})$ for $s=q^{-1/2}$, and $\mu/\nu\in (0, v^{-1})$ for $s=-q^{1/2}$. 
Then for the following choice of the normalizing constants
$$
(\widehat\sigma, \widehat\tau)=\begin{cases} \left(\dfrac{(\sqrt{\nu}-\sqrt{v\mu})^2}{1-v}\,,\dfrac{\left(v\mu\nu\right)^{1/6}\left(1-\sqrt{v\mu/\nu}\right)^{2/3}\left
(1-\sqrt{v\nu/\mu}\right)^{2/3}}{1-v}\right),& s=q^{-1/2},\\
\left(\dfrac{(\sqrt{\nu}-\sqrt{v\mu})^2}{1+v}\,, \dfrac{\left(v\mu\nu\right)^{1/6}\left(1-\sqrt{v\mu/\nu}\right)^{2/3}\left
(1+\sqrt{v\nu/\mu}\right)^{2/3}}{1+v}\right),&s=-q^{1/2},
\end{cases}
$$
the random variables $(\epsilon^{-2}\cdot \widehat \sigma -\ln\epsilon-\HT(M,N))$ have a weak limit as $\epsilon\to 0$; denoting it by $\xi$, the Laplace transform of $\exp(\xi)$ is given by \eqref{eq:Laplace-KPZ}. 
\end{theorem} 

The proof of this statement is very similar to the proof of Theorem \ref{th:ASEP-KPZ}, and we omit it. It rides on the heels of Theorems 6.1 and 6.3 of \cite{B-6v} in the same way as the proof of Theorem \ref{th:ASEP-KPZ} does with Theorem \ref{th:ASEP-TW} instead. In particular, the normalizing constants $\widehat\sigma$ and $\widehat \tau$ come from \cite{B-6v}. 

Theorem \ref{th:6v-KPZ} is closely related to the results of \cite{CT}, see, in particular, Theorem 2.8 there. Matching, however, does not seem entirely trivial, and we leave it to a future work.

\bigskip

Alexei Borodin:
\newline\indent Department of Mathematics, MIT, Cambridge, MA, USA;
\newline\indent Institute for Information Transmission Problems, Moscow, Russia
\newline\indent borodin@math.mit.edu

\medskip

Grigori Olshanski:
\newline\indent Institute for Information Transmission Problems, Moscow Russia;
\newline\indent National Research University Higher School of Economics, Moscow, Russia
\newline\indent olsh2007@gmail.com

\end{document}